\documentclass[11pt, UKenglish]{article}

\usepackage[UKenglish]{babel}

\usepackage{dsfont}
\usepackage{geometry}
\usepackage{graphics}
\usepackage{graphicx}
\usepackage{xcolor}
\usepackage{amsmath,amsfonts,amsthm,amssymb,framed,color,bbm}
\usepackage[utf8]{inputenc}

\usepackage{empheq}

\usepackage{pdfpages}

\usepackage{tikz}
\usetikzlibrary{matrix}

\usetikzlibrary{decorations.pathreplacing}

\theoremstyle{plain}
\newtheorem{thm}{Theorem}[section]
\newtheorem{prop}[thm]{Proposition}
\newtheorem{lemma}[thm]{Lemma}
\newtheorem{remark}[thm]{Remark}

\newtheorem{defn}[thm]{Definition}

\theoremstyle{definition}

\newcommand{\Z}{\mathbb{Z}}
\newcommand{\N}{\mathbb{N}}

\newcommand{\R}{\mathbb{R}}

\newcommand {\T} {\mathbb T} 

\newcommand{\diam}{\mathrm{diam}}

\newcommand{\es}{\emptyset}

\newcommand{\1}{\mathbbm{1}}

\newcommand{\vertiii}[1]{{\left\vert\kern-0.25ex\left\vert\kern-0.25ex\left\vert #1 
    \right\vert\kern-0.25ex\right\vert\kern-0.25ex\right\vert}}

\newcommand{\de}{\mathrm{d}}								

\begin{document}

\hyphenation{co-va-ri-ance cor-re-la-tion}

\title{Scaling limit and strict convexity of free energy for gradient models with non-convex potential}
\author{Susanne Hilger\footnote{Email: shilger@posteo.de}}
\date{}
\maketitle

\begin{abstract}
We consider gradient models on the lattice $\Z^d$. These models serve as effective models for interfaces and are also known as \textit{continuous Ising models}. The height of the interface is modelled by a random field with an energy which is a non-convex perturbation of the quadratic interaction. We are interested in the Gibbs measure with tilted boundary condition $u$ at inverse temperature $\beta$ of this model.
In 
\cite{AKM16}, 
\cite{Hil16}
and 
\cite{ABKM}
the authors show that for small tilt $u$ and large inverse temperature $\beta$ the surface tension is strictly convex, where the limit is taken on a subsequence. Moreover, it is shown that the scaling limit (again on a subsequence) is the Gaussian free field on the continuum torus. The method of the proof is a rigorous implementation of the renormalisation group method following a general strategy developed by Brydges and coworkers.

In this paper the renormalisation group analysis is extended from the finite-volume flow to an infinite-volume version to eliminate the necessity of the subsequence in the results in \cite{AKM16}, \cite{Hil16} and \cite{ABKM}.

\end{abstract}

\section{Introduction}

We analyse continuous Ising models which are effective models for random interfaces. Let $\Lambda \subset \Z^d$ be a finite subset of the lattice.  We consider fields $\varphi: \Lambda \to \R$ which can be interpreted as height variables of the interface. To each configuration $\varphi \in \R^{\Lambda}$ an energy $H_{\Lambda}(\varphi)$ is assigned This Hamiltonian is given by a potential $W:\R \to \R$ that only depends on discrete gradients of the field,
$$
H_{\Lambda}(\varphi) = \sum_{x \in \Lambda} \sum_{i = 1}^d W(\nabla_i \varphi (x)),
$$
where $\nabla_i \varphi (x) = \varphi(x+e_i) - \varphi(x)$ is the finite difference quotient on the lattice.
We impose tilted boundary conditions, namely
$$
\varphi(x) = \psi^u(x) \quad\mbox{for }x \in \partial\Lambda,
\quad \psi^u(x) = u \cdot x \mbox{ for } u \in \R^d.
$$
The finite-volume Gibbs measure with boundary condition $\psi^u$ at inverse temperature $\beta>0$ is then
$$
\gamma_{\beta,\Lambda}^{\psi^u} (\de \varphi) = \frac{1}{Z_{\beta,\Lambda}^{\psi^u}} \,\, e^{-\beta H_{\Lambda}(\varphi)} \prod_{x \in \Lambda} \de \varphi(x) \prod_{x \in \partial \Lambda} \delta_{\psi^u(x)}(\de \varphi (x)),
$$
where $Z_{\beta,\Lambda}^{\psi_u}$ is the partition function which normalizes the measure.

~\\
In the case of strictly convex, symmetric $W$ a lot is known about the behaviour of $\gamma_{\beta,\Lambda}^{\psi^u} (\de \varphi)$: The infinite-volume gradient Gibbs measure exists and is uniquely determined by the tilt, see \cite{FS97}. The long distance behaviour is described by the Gaussian free field (see \cite{NS97} and \cite{GOS01}) and the decay of the covariance is polynomial as in the massless Gaussian case (\cite{DD05}). Moreover the surface tension is strictly convex (DGI00).

~\\
The situation is not that clear for models with non-convex potentials.

A special class of gradient fields with non-convex potentials (log-mixture of centered Gaussians) is considered in \cite{BK07}. At tilt $u=0$, a phase transition is shown to happen at some critical value of inverse temperature $\beta_c$. This result demonstrates that one can expect neither the uniqueness of gradient Gibbs measures corresponding to a fixed tilt $u$ nor strict convexity of the surface tension. However, the scaling limit in this case is still the Gaussian free field, as shown in \cite{BS11}.

For a class of gradient models where the potential is a small non-convex perturbation of a strictly convex one, \cite{CDM09}shows strict convexity of the surface tension at high temperature. For the same class in the same temperature regime, in \cite{CD12} it is shown that for any $u$ there exists a unique ergodic, shift-invariant gradient Gibbs measure . Moreover, the measure scales to the Gaussian free field and the decay of the covariance is algebraic as above.

~\\
The complementary temperature regime is considered in \cite{AKM16}. The authors consider potentials which are small perturbations of the quadratic one, the perturbation chosen such that it does not disturb the convexity at the minimum of the potential. For small tilt $u$ and large inverse temperature $\beta$ they prove strict convexity of the surface tension obtained as a limit of a subsequence of $(N_l)_{l \in \N}$, where $L^N$ is the side length of the box $\Lambda$, and relying on a quite restrictive lower bound on $W$, namely
$$
W(s) \geq (1-\epsilon) s^2
$$
for a small $\epsilon$.

In the same setting the paper \cite{Hil16} shows that there is $q \in \R^{d \times d}_{\text{sym}}$ small, such that the scaling limit is the Gaussian free field on $\T^d$ with covariance $\mathcal{C}^q_{\T^d}$, where
$$
 \left( \mathcal{C}_{\T^d}^q \right)^{-1}
 = - \sum_{i,j=1}^d \left( \delta_{ij} + q_{ij} \right) \partial_i \partial_j,
$$
 and that a "smoothed" covariance decays algebraically. The convergences are on a subsequence.

~\\
In \cite{ABKM} the class of potentials is widened to such which satisfy less restrictive bounds on the potential, namely
$$
W(s) \geq \epsilon s^2,
$$
and to vector-valued fields and finite-range instead of only nearest-neighbour interaction. The last two improvements are of interest for the application in nonlinear elasticity. The authors show that the surface tension is strictly convex and that the scaling limit is the Gaussian free field on the torus. Unfortunately, all convergences are still on a subsequence.

~\\
The setting in this paper is similar to the one from \cite{ABKM}: We restrict to small tilts and large inverse temperature and use the same smallness condition on the potential. For the sake of simplicity we formulate our results and proofs for scalar-valued fields and nearest-neighbour interaction. We show that the necessity for the subsequence in the statements about the surface tension and the scaling limit can be removed.

~\\
The proof builds on a rigorous renormalisation group approach for the partition function as developed by Bauerschmidt, Brydges and Slade in a series of papers (\cite{BS1},\cite{BS2}, \cite{BS3}, \cite{BS4}, \cite{BS5}). This approach is developed for the model at hand in \cite{AKM16} and improved in \cite{ABKM}. We augment the technique in the following direction: we extend the finite-volume flow apparent in the renormalisation group method to infinite volume. This enables us to get rid of the restriction on the subsequence.

\paragraph{Structure of the paper}
In Section \ref{sec:SettingResults}, gradient models are introduced and the main results concerning the scaling limit (Theorem \ref{Thm:ScalingLimit}) and the strict convexity of the surface tension (Theorem \ref{Thm:SurfaceTension}) are stated. Furthermore, a technical theorem on which the proofs of these results are based is formulated (Theorem \ref{Thm:RepresentationPartitionFunction}). THe technical theorem contains a representation of the generating partition function and provides straightforward proofs of the main results.

Finally, Section \ref{sec:RG-analysis_BulkFlow} contains the proof of the first technical result, Theorem \ref{Thm:RepresentationPartitionFunction}. The proof is by renormalisation group analysis which closely follows \cite{ABKM}. To improve the convergence results in \cite{ABKM}, the method is extended from finite-volume to infinite-volume flows. This extension is explained in \cite{BS5} for the $\varphi^4$-model and adapted to gradient models in this paper.

\paragraph{Notations}
Throughout the whole paper we will use the following notations.
\begin{itemize}
	\item $C_c^{\infty}$ will denote the set of smooth, compactly supported functions.
	\item Partial derivatives will be denoted by $\partial_s$ instead of $\frac{\partial}{\partial s}$.
	\item The symbol $\partial_i$ will be used for usual derivatives, in contrast to $\nabla_i$ for discrete finite differences.
	\item $C^r$ denotes the set of $r$-times differential functions.
	\item $\R^{d \times d}_{{\text{sym}}}$ denotes the set of $d \times d$ symmetric matrices.
	\item The Kronecker-delta $\delta_{ij}$ is $1$ if $i=j$ and $0$ else.
	\item The indicator function $\1_z$ is given by $\1_z = 1$ if condition $z$ is satisfied and $\1_z = 0$ otherwise.
	\item The symbol $C$ will mostly denote a positive constant whose value is allowed to change in a chain of inequalities from line to line.
\end{itemize}

\section{Setting and results} \label{sec:SettingResults}

We start by describing gradient models and their finite-volume Gibbs distributions and stating the main results, namely the scaling limit of the measure in Theorem~\ref{Thm:ScalingLimit}, strict convexity of the surface tension in Theorem \ref{Thm:SurfaceTension}.

Then we state a technical key theorem (Theorem \ref{Thm:RepresentationPartitionFunction}), which contains a powerful representation of the normalisation constant of the Gibbs measure. From this representation the proofs of the main results can be deduced straightforwardly.

\subsection{Gradient models}\label{sec:GradientModels}

Fix an integer $L \geq 3$ and a dimension $d \geq 2$. Let $\T_N = \left( \Z / L^N \Z \right)^d$ be the $d$-dimensional discrete torus of side length $L^N$ where $N$ is a positive integer. We equip $\T_N$ with the quotient distances $| \cdot |$ and $| \cdot |_{\infty}$ induced by the Euclidean and maximum norm respectively. The torus can be represented by the cube
$$
\Lambda_N = \left\lbrace x \in \Z^d: |x|_{\infty} \leq \frac{1}{2} \left(L^N - 1 \right) \right\rbrace
$$
of side length $L^N$ once it is equipped with the metric
$$
|x-y|_{\mathrm{per}} = \inf \left\lbrace |x-y+k|_{\infty}: k \in \left( L^N \Z \right)^d \right\rbrace.
$$

 Define the space of $m$-component fields on $\Lambda_N$ as 
$$
\mathcal{V}_N = \lbrace \varphi: \Lambda_N \rightarrow \R^m \rbrace = \left(\R^m \right)^{\Lambda_N}.
$$
Since we will consider shift invariant energies, we are only interested in gradient fields on $\mathcal{V}_N$. Gradient fields can be described by elements in $\mathcal{V}_N /_{\lbrace\text{constants}\rbrace}$, or, equivalently, by usual fields with vanishing average
$$
\chi_N = \bigg\lbrace \varphi \in \mathcal{V}_N: \sum_{x \in \Lambda_N} \varphi(x) = 0 \bigg\rbrace.
$$
We equip $\chi_N$ with a scalar product via
$$
(\varphi,\psi) = \sum_{x \in \Lambda_N} \varphi(x) \psi(x).
$$
Let $\lambda_N$ be the $m\left(L^{Nd}-1\right)$-dimensional Hausdorff measure on $\chi_N$. Let $e_i$, $i= 1, \ldots, d$, be the standard unit vectors in $\Z^d$. Then the discrete forward and backward derivatives are defined by
\begin{align*}
\left( \nabla_i \varphi \right)_s(x) 
= \varphi_s (x + e_i) - \varphi_s(x), 
\quad i \in \lbrace 1, \ldots, d \rbrace,
\quad s \in \lbrace 1, \ldots, m \rbrace,
\\
\left( \nabla_i^* \varphi \right)_s (x) 
= \varphi_s (x - e_i) - \varphi_s(x), 
\quad i \in \lbrace 1, \ldots, d \rbrace,
\quad s \in \lbrace 1, \ldots, m \rbrace.
\end{align*}

Let $A \subset \Z^d$ be a finite set with range $R_0 = \diam_{\infty} (A)$ and let $U: \left(\R^m\right)^A \rightarrow \R$ be a finite-range potential which is invariant with respect to translations in $\R^m$, i.e., $U(\psi) = U(\tau_a \psi)$ for any $\psi \in \left( \R^m \right)^A$ with $\left( \tau_a \psi \right)_s (x) = \psi_s(x) + a_s$, $a \in \R^m$.

We study a class of random gradient fields defined in terms of a Hamiltonian
\begin{align*}
H_N(\varphi)
= \sum_{x \in \Lambda_N} U \left( \varphi_{\tau_x(A)} \right),
\quad
\tau_x(A) = A+x = \lbrace y: y-x \in A \rbrace
\quad \text{for } x \in \Lambda_N.
\end{align*}

We equip the space $\chi_N$ with the $\sigma$-algebra $\mathfrak{B}_{\chi_N}$ induced by the Borel-$\sigma$-algebra with respect to the product topology, and use $\mathcal{M}_1(\chi_N) = \mathcal{M}_1(\chi_N, \mathfrak{B}_{\chi_N} )$ to denote the set of probability measures on $\chi_N$.

The finite-volume gradient Gibbs measure $\gamma_{N,\beta} \in \mathcal{M}_1(\chi_N)$ at inverse temperature $\beta$ is defined as
\begin{align*}
\gamma_{N,\beta}(\de \varphi) = \frac{1}{Z_{N,\beta}} e^{-\beta H_N(\varphi)} \lambda_N(\de \varphi)
\end{align*}
with partition function
\begin{align*}
Z_{N,\beta} = \int_{\chi_N} e^{-\beta H_N(\varphi)} \lambda_N(\de \varphi).
\end{align*}

We implement suitable boundary conditions following the Funaki-Spohn-trick introduced in \cite{FS97}. Given a linear map $F: \R^d \rightarrow \R^m$, we define the Hamiltonian $H_N^F$ on $\Lambda_N$ with the external deformation $F$ by
\begin{align*}
H_N^F(\varphi) = \sum_{x \in \Lambda_N}U\left( (\varphi + F)_{\tau_x(A)} \right).
\end{align*}
Consequently, the finite-volume gradient Gibbs measure $\gamma_{N,\beta}^F$ with deformation $F$ is defined~as
\begin{align*}
\gamma_{N,\beta}^F(\de \varphi) = \frac{1}{Z_{N,\beta}(F)} e^{-\beta H_N^F(\varphi)} \lambda_N(\de\varphi),
\end{align*}
where $Z_{N,\beta}(F)$ is the normalisation constant. A useful generalisation of the partition function with a source term $f \in \mathcal{V}_N$ is given by the generating functional
\begin{align}
Z_{N,\beta}(F,f) = \int_{\chi_N} e^{-\beta H_N^F(\varphi)+(f,\varphi)} \lambda_N(\de \varphi).
\label{eq:gen_func}
\end{align}

We can rewrite the model as \textit{generalized gradient model} as it is done in detail in \cite{ABKM}, Section 2.2. Let $Q_{R_0} = \lbrace 0, \ldots, R_0 \rbrace^d$. We introduce the $m$-dimensional space of shifts
$$
\mathcal{V}_{Q_{R_0}} = \left\lbrace (a, \ldots, a) \in \left( \R^m \right)^{Q_{R_0}}: a \in \R^m \right\rbrace
$$ 
and its orthogonal complement $\mathcal{V}^{\bot}_{Q_{R_0}}$ in $\left( \R^m \right)^{Q_{R_0}}$. Furthermore, let
$$
\mathcal{I}_{R_0} =
\left\lbrace \alpha \in \N_0^d \setminus \lbrace (0,\ldots , 0) \rbrace : |\alpha|_{\infty} \leq R_0 \right\rbrace,
\quad
\text{and}
\quad
\mathcal{G}_{R_0} = \left( \R^m \right)^{\mathcal{I}_{R_0}}.
$$
We will also need a more general index set $\mathcal{I}$ given by
$$
\lbrace e_i \in \R^d: 1 \leq i \leq d \rbrace
\subset \mathcal{I}
\subset \mathcal{I}_{R_0}
$$
and the corresponding vector space $\mathcal{G} = \left( \R^m \right)^{\mathcal{I}}$.
We define the \textit{extended gradient} $D\varphi(x)$ as the vector $\left( \nabla^{\alpha} \varphi(x) \right)_{\alpha \in \mathcal{I}} \in \mathcal{G}$ where $\nabla^{\alpha}\varphi(x) = \prod_{j=1}^d \nabla_j^{\alpha(j)} \varphi(x)$.

Lemma 2.1 in \cite{ABKM} states that for any $U: \mathcal{V}^{\bot}_{Q_{R_0}} \rightarrow \R$ there is $\mathcal{U}: \mathcal{G}_{R_0} \rightarrow \R$ such that $\mathcal{U}(D \psi (0)) = U(\psi)$ for any $\psi \in \mathcal{V}^{\bot}_{Q_{R_0}}$. Moreover, the linear deformation $F: \R^d \rightarrow \R^m$ can be identified with the element $\bar{F}= DF(x) \in \mathcal{G}$ (for any $x \in \R^d$). Thus we get $U(\psi + F) = \mathcal{U}(D \psi(0) + \bar{F})$ for any $\psi \in \mathcal{V}_{Q_{R_0}}^{\bot}$ leading to an alternative expression for the Hamiltonian $H_N^F(\varphi)$,
\begin{align}
H_N^F(\varphi)
= \sum_{x \in \Lambda_N} U \left( \varphi_{\tau_x(A)} + F \right)
= \sum_{x \in \Lambda_N} \mathcal{U}(D \varphi(x) + \bar{F}).
\end{align}
In the following we will use the formulation of the Hamiltonian relying on extended gradients.

\subsection{Main results}

We give improved versions of Theorem 2.9 in \cite{ABKM} (which was firstly proven in \cite{Hil16} with stronger assumptions on the potential $\mathcal{U}$) and Theorem 2.6 in \cite{ABKM}. The improvement consists in the removal of the need for a subsequence $(N_l)_l$.

As in \cite{ABKM}, let $\mathbf{Q}: \mathcal{G} \rightarrow \mathcal{G}$ be a symmetric positive linear operator and $\mathcal{Q}: \mathcal{G} \rightarrow \R$ the corresponding quadratic form $\mathcal{Q}(z) = (z, \mathbf{Q}z)$. Set
\begin{align}
\mathcal{Q}_{\mathcal{U}}(z)
= D^2 \mathcal{U}(0)(z,z).
\end{align}

~\\
We impose the following assumptions on the potential $\mathcal{U}$:
\begin{align}
\begin{cases}
&\mbox{Let } r_0, r_1 \in \N_0, \, \mathcal{U} \in C^{r_0+r_1}(\mathcal{G}), 
\\ & \quad\quad
\text{ and } \omega_0 |z|^2 \leq \mathcal{Q}_{\mathcal{U}}(z) \leq \omega_0^{-1} |z|^2 \text{ for some } \omega_0 \in (0,1). \nonumber
\\
&\mbox{Let } 0 < \omega < \frac{\omega_0}{8}  \mbox{ and suppose that } 
\mathcal{U}(z) - D \mathcal{U}(0)z - \mathcal{U}(0) \geq \omega |z|^2 \text{ for all } z \in \mathcal{G}
\tag{$\star$}
\label{eq:AssumptionsW}
\\
&  \mbox{ and}
\lim_{t \rightarrow \infty}t^{-2} \ln \Psi (t) = 0
\\ & \quad\quad
 \mbox{ where }
\Psi(t) = \sup_{|z|\leq t} \sum_{3 \leq |\alpha| \leq r_0+r_1} \frac{1}{\alpha !} \lvert \partial^{\alpha} \mathcal{U}(z)\rvert.
\nonumber
\end{cases}
\end{align}

~\\
Let $\T^d = \left( \R/\Z \right)^d$ be the continuum torus, $q: \R^{d \times m} \rightarrow \R^{d \times m}$ be symmetric, and $\mathcal{C}^{q}_{\T^d}$ be the inverse of the elliptic partial differential operator $\mathcal{A}^{q}_{\T^d}$,
$$
\mathcal{C}^{q}_{\T^d} = \left( \mathcal{A}^{q}_{\T^d} \right)^{-1},
\quad
\mathcal{A}^{q}_{\T^d}
 = - \sum_{i,j=1}^d \left( \left(\mathbf{Q}_{\mathcal{U}}\right)_{ij} - q_{ij} \right) \partial_j \partial_i,
$$
which acts on the space of all functions $f \in W^{1,2}(\T^d)$ with mean zero.

The following theorem states that the Laplace transform of $\gamma_{N,\beta}^F$ converges to the Laplace transform of the Gaussian free field $\mu_{\mathcal{C}^q_{\T^d}}$ on the continuum torus with covariance $\mathcal{C}^q_{\T^d}$ as the lattice spacing tends to zero in a suitably scaled way. The convergence is not restricted to a subsequence as it is needed in a similar statement in \cite{ABKM}.

\begin{thm}[Scaling limit]\label{Thm:ScalingLimit}
Fix the spatial dimension $d$, the number of components $m$, the range of interaction $R_0$, $\omega_0 \in (0,1)$, $r_0 \geq 3$, $r_1 \geq 0$ and let $\mathcal{U}$ satisfy \eqref{eq:AssumptionsW}.
Then there is $L_0$ such that for all integers $L \geq L_0$ there is $\delta > 0$ and $\beta_0>0$ with the following property.
For all $F \in B_{\delta}(0)$  and $\beta \geq \beta_0$ there is $q \in \R^{(d \times m) \times (d \times m)}_{\mathrm{sym}}$ such that for any $f \in C_c^{\infty}\left(\T^d\right)$ satisfying $\int f = 0$ and $f_N(x)= L^{-N \frac{d-2}{2}}f\left(L^{-N}x\right)$ for $x \in \Lambda_N$,
\begin{align*}
\lim_{N \rightarrow \infty} \mathbb{E}_{\gamma_{N,\beta}^F}(e^{(f_N,\cdot)})
=
\lim_{N \rightarrow \infty} \frac{Z_{N,\beta}(F,f_N)}{Z_{N,\beta}(F,0)}
= e^{\frac{1}{2 \beta}\left(f,\,\mathcal{C}^q_{\T^d} f\right)}.
\end{align*}
\end{thm}

~\\
Let us denote
\begin{align}
W_{N,\beta}(F) = - \frac{1}{\beta L^{Nd}} \ln Z_{N,\beta}(F,0).
\end{align}
The \textit{free energy} can be written as
\begin{align}
W_{\beta}(F) = \lim_{N \rightarrow \infty} W_{N,\beta}(F).
\end{align}

In the following Theorem we show strict convexity of $W_{\beta}$ for small deformations $F$ and small temperature $\beta^{-1}$. The convergence $W_{N,\beta}(F) \rightarrow W_{\beta}(F)$ is not restricted to a subsequence as it is done in a similar statement in \cite{ABKM}.

\begin{thm}[Strict convexity of free energy]\label{Thm:SurfaceTension}
Fix the spatial dimension $d$, the number of components $m$, the range of interaction $R_0$, $\omega_0 \in (0,1)$, $r_0 \geq 3$, $r_1 \geq 2$ and let
 $\mathcal{U}$ satisfy \eqref{eq:AssumptionsW}.
Then there is $L_0$ such that for all integers $L \geq L_0$ there is $\delta > 0$ and $\beta_0$ with the following property.
For all $F \in B_{\delta}(0)$ and $\beta \geq \beta_0$ there is $q \in \R^{(d \times m) \times ( d \times m)}_{\text{sym}}$ such that
for any $N$ the free energy $W_{N,\beta}:B_{\delta}(0)\rightarrow \R$ is in $C^{r_1}$ and uniformly convex. Moreover, the limit $W_{\beta}(F)$ is uniformly convex in $B_{\delta}(0)$.
\end{thm}


\begin{remark} \label{rem:extension_main_thm}

One can state the assumptions \eqref{eq:AssumptionsW} on the potential $\mathcal{U}$ in a more general form allowing a bigger class. We will comment on this again in the next section, see Lemma~\ref{Lemma:From_K_to_V} and Remark \ref{rem:weak_ass_V}. For the sake of simplicity we decided to state the main results with assumptions \eqref{eq:AssumptionsW}.

\end{remark}

\subsection{Key theorem and proofs of the main results}

The goal of this section is the formulation of a technical key theorem (Theorem \ref{Thm:RepresentationPartitionFunction}), which states a powerful representation of the generating functional of the model. The proof of this theorem is obtained by a subtle renormalisation group analysis which is an extension of the corresponding proof in \cite{ABKM} and  will be carried out in Section \ref{sec:RG-analysis_BulkFlow}.

\subsubsection{Reformulation of $Z_{N,\beta}(F,f)$}

Similar to \cite{ABKM}, we define $\overline{\mathcal{U}}(z,F)$ by
\begin{align*}
\overline{\mathcal{U}}(z,F) = \mathcal{U}(z+\bar{F}) -  \mathcal{U}(\bar{F}) - D  \mathcal{U}(\bar{F})z
- \frac{\mathcal{Q}_{\mathcal{U}}(z)}{2}.
\end{align*}
We can write the generating functional $Z_{N,\beta}(F,f)$ from \eqref{eq:gen_func} in the form
\begin{align*}
Z_{N,\beta}(F,f) &=
e^{- \beta L^{Nd}  \mathcal{U}(\bar{F})} 
\\ & \quad \times
\int_{\chi_N} e^{(f,\varphi)} e^{- \beta \sum_{x \in \Lambda_N}
\left(
 \overline{\mathcal{U}}(D\varphi(x), F) + \frac{\mathcal{Q}_{ \mathcal{U}}(D\varphi(x))}{2}
\right)
}
 \lambda_N(\de \varphi).
\end{align*}
Let
\begin{align}
\mu_{\beta}(\de\varphi) = 
\frac{1}{Z_{N,\beta}^{\mathcal{Q}_ {\mathcal{U}}}} 
e^{-\frac{\beta}{2} \sum_{x \in \Lambda_N}
\mathcal{Q}_{ \mathcal{U}}(D\varphi(x))
} \lambda_N(\de\varphi)
\end{align}
be the Gaussian measure at inverse temperature $\beta$ with corresponding normalisation factor
\begin{align}
Z_{N,\beta}^{\mathcal{Q}_{ \mathcal{U}}} 
= \int_{\chi_N} e^{-\frac{\beta}{2} \sum_{x \in \Lambda_N}
\mathcal{Q}_{ \mathcal{U}}(D\varphi(x))
} \lambda_N(\de\varphi).
\label{eq:10}
\end{align}
Consequently,
\begin{align*}
Z_{N,\beta}(F,f)
=
e^{- \beta L^{Nd}  \mathcal{U}(\bar{F}) } 
Z_{N,\beta}^{\mathcal{Q}_{ \mathcal{U}}}
\int_{\chi_N} 
e^{(f,\varphi)}
 e^{- \beta \sum_{x \in \Lambda_N}
 \overline{ \mathcal{U}} \left(
\frac{D \varphi(x)}{\sqrt{\beta}}, F
\right)
} \mu_{\beta}(\de\varphi).
\end{align*}
Now we rescale the field by $\sqrt{\beta}$ and introduce the \textit{Mayer function} $\mathcal{K}_{F,\beta, \mathcal{U}}$,
\begin{align}
\mathcal{K}_{F,\beta, \mathcal{U}}(z) = 
e^{-\beta 
\overline{ \mathcal{U}}\left( \frac{z}{\sqrt{\beta}}, F \right)
} 
- 1.
\label{eq:zero_perturbation_K}
\end{align}
We can express the partition function $Z_{N,\beta}(F,f)$ in terms of the polymer expansion:
\begin{align*}
Z_{N,\beta}(F,f)
&=
e^{- \beta L^{Nd}  \mathcal{U}(\bar{F}) } Z_{N,\beta}^{\mathcal{Q}_{ \mathcal{U}}}
\int_{\chi_N}
 e^{\left(f,\frac{\varphi}{\sqrt{\beta}}\right)}
e^{- \beta \sum_{x \in \Lambda_N} \overline{ \mathcal{U}} \left(\frac{D \varphi(x)}{\sqrt{\beta}}, F \right)} \mu_1(\de\varphi)
\\
&=
e^{- \beta L^{Nd}  \mathcal{U}(\bar{F})} Z_{N,\beta}^{\mathcal{Q}_{ \mathcal{U}}}
\int_{\chi_N} 
e^{\left(f,\frac{\varphi}{\sqrt{\beta}}\right)}
\prod_{x \in \Lambda_N} \left( 1 + \mathcal{K}_{F,\beta, \mathcal{U}}(D \varphi(x)) \right) \mu_1(\de\varphi)
\\
&=
e^{- \beta L^{Nd}  \mathcal{U}(\bar{F})} Z_{N,\beta}^{\mathcal{Q}_{ \mathcal{U}}}
\int_{\chi_N} 
e^{\left(f,\frac{\varphi}{\sqrt{\beta}}\right)}
\sum_{X \subset \Lambda_N} \prod_{x \in X} 
\mathcal{K}_{F,\beta, \mathcal{U}}(D \varphi(x)) \mu_1(\de\varphi).
\end{align*}
The integral in the last expression gives the perturbative contribution
\begin{align}
\mathcal{Z}_{N,\beta}\left( F,\frac{f}{\sqrt{\beta}} \right) =
\int_{\chi_N} e^{\left(\frac{f}{\sqrt{\beta}},\varphi\right)}
\sum_{X \subset \Lambda_N} \prod_{x \in X} \mathcal{K}_{F,\beta, \mathcal{U}}(D \varphi(x)) \mu_1(\de\varphi).
\label{eq:perturbative_part}
\end{align}
In summary, we obtain the representation
\begin{align*}
Z_{N,\beta}(F,f)
=
e^{- \beta L^{Nd}  \mathcal{U}(\bar{F})} Z_{N,\beta}^{\mathcal{Q}_{ \mathcal{U}}}
\,
\mathcal{Z}_{N,\beta}\left( F,\frac{f}{\sqrt{\beta}} \right).
\end{align*}

We introduce a space for the perturbation $\mathcal{K}_{F,\beta,\mathcal{U}}$.
Let $\zeta \in (0,1)$. For $r_0 \geq 3$ we define the Banach space $\mathbf{E}_{\zeta,\mathcal{Q}_{\mathcal{U}}}$ consisting of functions $\mathcal{K}:\mathcal{G} \rightarrow\R$ such that the following norm is finite
\begin{align*}
\Vert \mathcal{K} \Vert_{\zeta, \mathcal{Q}} 
= \sup_{z \in \R^d} \sum_{|\alpha| \leq r_0} \frac{1}{\alpha !} \vert \partial^{\alpha}\mathcal{K}(z) \vert e^{-\frac{1}{2}(1- \zeta)  \mathcal{Q}_{\mathcal{U}}(z)
}.
\end{align*}

Let us generalise the expression for the perturbative part of the partition function in \eqref{eq:perturbative_part} to general positive definite quadratic form $ \mathcal{Q}: \mathcal{G} \rightarrow \R$ and to arbitrary $\mathcal{K} \in \mathbf{E}_{\zeta, \mathcal{Q}}$ from the rather explicit $\mathcal{K}_{F,\beta, \mathcal{U}}$ in \eqref{eq:zero_perturbation_K}. Namely, let
\begin{align}
\mathcal{Z}_{N}\left( \mathcal{K},  \mathcal{Q},f \right) =
\int_{\chi_N} e^{\left(f,\varphi\right)}
\sum_{X \subset \Lambda_N} \prod_{x \in X} \mathcal{K}(D \varphi(x)) \mu_{ \mathcal{Q}}(\de\varphi)
\label{eq:expression_PF}
\end{align}
with the Gaussian measure
\begin{align*}
\mu_{ \mathcal{Q}}(\de\varphi)
= \frac{1}{Z_{N,\beta}^{\mathcal{Q}}}
e^{-\frac{1}{2} \sum_{x \in \Lambda_N} \mathcal{Q}(D\varphi(x))} \lambda_N(\de\varphi).
\end{align*}
In the next subsection we state a representation for \eqref{eq:expression_PF} under some conditions on $\mathcal{K}$ and $\mathcal{Q}$ and conclude the proofs for Theorems \ref{Thm:ScalingLimit} and \ref{Thm:SurfaceTension}.

\subsubsection{Representation of $Z_{N,\beta}(F,f)$ and conclusions}

Let us introduce $\mathcal{C}^q_{\Lambda_N} = \left( \mathcal{A}_{\Lambda_N}^{q} \right)^{-1}$ for a symmetric map $q: \R^{d \times m} \rightarrow \R^{d \times m}$, where
$$
\mathcal{A}^{q}_{\Lambda_N}:
\chi_N \rightarrow \chi_N,
\quad
\mathcal{A}^{q}_{\Lambda_N}
 = \sum_{i,j=1}^d \left( \mathbf{Q}_{ij} - q_{ij} \right) \nabla_j^* \nabla_i.
$$
We use $\Vert q \Vert$ to denote the operator norm of $q$ viewed as an operator on $\R^{d \times m}$ equipped with the $l_2$ metric.
If $q$ is small, $\Vert q \Vert \leq \frac{1}{2}$, we can define a Gaussian measure $\mu_{\mathcal{C}^{q}_{\Lambda_N}}$ on $\chi_N$ with covariance $\mathcal{C}^{q}_{\Lambda_N}$,
$$
\mu_{\mathcal{C}^{q}_{\Lambda_N}} (\de \varphi) = \frac{1}{Z_N^{(q)}} e^{-\frac{1}{2}\left( \varphi, \mathcal{A}^{q}_{\Lambda_N} \varphi \right)} \de \lambda_N(\varphi).
$$

~\\
The following theorem states that the perturbative contribution $\mathcal{Z}_{N}\left( \mathcal{K}, \mathcal{Q},f \right)$ in \eqref{eq:expression_PF} can be written as the product of a rather explicit term and a term which is almost $1$, the error being exponentially decreasing in $N$ if $\mathcal{K}$ is small enough. This result is the key ingredient for the proofs of Theorem \ref{Thm:ScalingLimit} and Theorem \ref{Thm:SurfaceTension}. The proof is a subtle renormalisation group~(RG) analysis established in \cite{AKM16} and reviewed and extended in Section~\ref{sec:RG-analysis_BulkFlow}.

\begin{thm}[Representation of the partition function]\label{Thm:RepresentationPartitionFunction}
Fix $\zeta, \eta \in (0,1)$ and let $\mathcal{Q}$ be a quadratic form on $\mathcal{G}$ satisfying $\omega_0 |z|^2 \leq \mathcal{Q}(z) \leq \omega_0^{-1}|z|^2$. There is $L_0$ such that for all integers $L \geq L_0$ there is $\epsilon_0 > 0$ with the following properties. There exist smooth maps
$$
e: B_{\epsilon_0}(0) \subset \mathbf{E}_{\zeta} \rightarrow \R,
\quad
q: B_{\epsilon_0}(0) \subset \mathbf{E}_{\zeta} \rightarrow \R^{(d \times m) \times (d \times m)}_{\mathrm{sym}},
$$
and, for any $N \in \N$, a smooth map $Z_N: B_{\epsilon_0}(0) \times \chi_N \rightarrow \R$ (with bounds on the derivatives which are uniform in $N$)
 such that for any $f \in \chi_N$ and $\mathcal{K} \in B_{\epsilon_0}(0)$
 the following representation holds:
\begin{align*}
\mathcal{Z}_{N}(\mathcal{K}, \mathcal{Q},f)
=  e^{\frac{1}{2}\left(f,\mathcal{C}_{\Lambda_N}^{q(\mathcal{K})}f\right)}
\frac{Z_N^{(q(\mathcal{K}))}}{Z_N^{(0)}}
e^{- L^{Nd} e(\mathcal{K}) }
Z_N\left(\mathcal{K},\mathcal{C}_{\Lambda_N}^{q(\mathcal{K})}f\right).
\end{align*}
If $f(x) = g_N(x)-c_N$, $g_N(x)=L^{-N\frac{d+2}{2}}g(L^{-N}x)$ for $g \in C_c^{\infty}(\T^d)$ with $\int g =0$, $c_N$ such that $\sum_{x \in \T_N}f(x) = 0$, then there is a constant $C$ which is independent of $N$ such that the remainder $Z_N(\mathcal{K})$ satisfies the estimate
\begin{align*}
 \left\vert
Z_N \left(\mathcal{K}, \mathcal{C}_{\Lambda_N}^{q(\mathcal{K})}f\right) - 1
\right\vert
\leq C \eta^N.
\end{align*}
\end{thm}

Notice that the condition on $f$ includes the case $f \equiv 0$.

\begin{remark}
This statement is similar to Theorem 11.1 in \cite{ABKM} with the key difference that in \cite{ABKM} the quantities $e(\mathcal{K})$ and $q(\mathcal{K})$ depend on the size of the torus, i.e., on $N$, and here they are independent of $N$. This improvement is obtained by introducing a global flow (see Section \ref{subsection:Inf_vol_bulk_flow}). As a consequence, there is no subsequence needed in Theorems \ref{Thm:ScalingLimit} and \ref{Thm:SurfaceTension}.

\end{remark}

Proposition 2.4 in \cite{ABKM} provides conditions on $\mathcal{U}$ such that $\mathcal{K}_{F,\beta,\mathcal{U}} \in B_{\rho}(0) \subset \mathbf{E}_{\zeta}$ for any $\rho>0$ is satisfied. We cite the proposition in the following lemma.

\begin{lemma}
\label{Lemma:From_K_to_V}
Let $r_0 \geq 3$ and $r_1 \geq 0$ be integers and assume that $\mathcal{U}$ satisfies \eqref{eq:AssumptionsW}.
Then there exist $\tilde{\zeta}$, $\delta_0 > 0$, $C_1$ and $\theta > 0$ such that for all $\delta \in \left(0,\delta_0\right]$ and for all $\beta \geq 1$ the map 
$$
B_{\delta}(0) \ni F \mapsto \mathcal{K}_{F,\beta,\mathcal{U}} \in \mathbf{E}_{\tilde{\zeta},\mathcal{Q}_{\mathcal{U}}}
$$
is $C^{r_1}$ and satisfies
\begin{align}
&\Vert \mathcal{K}_{F,\beta,\mathcal{U}} \Vert_{\tilde{\zeta},\mathcal{Q}_{\mathcal{U}}}
\leq
C_1 \left( \delta + \beta^{-\frac{1}{2}} \right)
\quad \text{and} \quad
\sum_{|\gamma|\leq r_1} \frac{1}{\gamma !} \Vert \partial_u^{\gamma} \mathcal{K}_{F,\beta,\mathcal{U}} \Vert_{\tilde{\zeta},\mathcal{Q}_{\mathcal{U}}} \leq \theta.
\label{eq:Bound_K}
\end{align}
In particular, given $\rho>0$, there exist $\delta > 0$ and $\beta_0 \geq 1$ such that for all $\beta \geq \beta_0$ and all $F \in B_{\delta}(0)$ we have
$$
\Vert \mathcal{K}_{F,\beta,\mathcal{U}} \Vert_{\tilde{\zeta},\mathcal{Q}_{\mathcal{U}}} \leq \rho
$$ 
and the bound on the derivatives in \eqref{eq:Bound_K} holds.
\end{lemma}

\begin{remark}\label{rem:weak_ass_V}
As noted in the previous section we can state more general assumptions on the potential $\mathcal{U}$ than \eqref{eq:AssumptionsW}. Namely, it is enough to assume the smallness condition on the Mayer function $\mathcal{K}$, $\Vert \mathcal{K}_{F,\beta,\mathcal{U}} \Vert_{\tilde{\zeta},\mathcal{Q}_{\mathcal{U}}} \leq \rho$. Then the main theorems can be applied for every $\mathcal{U}$ such that its Mayer function satisfies the bound.
\end{remark}

The proofs of Theorems \ref{Thm:ScalingLimit} and \ref{Thm:SurfaceTension} are straightforward consequences of the representation of the partition function in Theorem \ref{Thm:RepresentationPartitionFunction}.

\begin{proof}[Proof of Theorem \ref{Thm:ScalingLimit}]
The proof may be handled in the very same way as in \cite{Hil16} or \cite{ABKM} but without the need for taking a subsequence. We review the main arguments.

Let $\tilde{\zeta}$ be the parameter from Lemma \ref{Lemma:From_K_to_V}, and let $L_0$ and $\epsilon_0$ the corresponding parameters from Theorem \ref{Thm:RepresentationPartitionFunction}. Then, by Lemma \ref{Lemma:From_K_to_V}, there is $\delta > 0$ and $\beta_0 \geq 1$ such that for all $\beta \geq \beta_0$ and $F \in B_{\delta}(0)$ we have $\mathcal{K}_{F,\beta,\mathcal{U}} \in B_{\epsilon_0}(0) \subset \mathbf{E}_{\tilde{\zeta},\mathcal{Q}_{\mathcal{U}}}$. Fix $f \in \chi_N$. By Theorem~\ref{Thm:RepresentationPartitionFunction}, the function $\mathcal{Z}_{N,\beta}(F,f)$ can be written as an explicit term multiplied by a perturbation $Z_N(\mathcal{K}_{F,\beta,\mathcal{U}})$.

Let $f_N$ be as in the assumptions of the theorem. Define
$$
\tilde{f}_N = f_N - c_N,
\quad
c_N \text{ such that } \sum_{x \in \T_N} \tilde{f}_N(x) =0.
$$
Then $\tilde{f}_N \in \chi_N$.
Since $(c_N,\varphi)=0$ for all $\varphi \in \chi_N$,
$$
\mathbb{E}_{\gamma_{N,\beta}^F}\left( e^{(f_N,\varphi)} \right)
=
\mathbb{E}_{\gamma_{N,\beta}^F}\left( e^{(\tilde{f}_N,\varphi)} \right),
$$
and we can use Theorem \ref{Thm:RepresentationPartitionFunction} to rewrite
\begin{align*}
\mathbb{E}_{\gamma_{N,\beta}^F}\left( e^{(f_N,\varphi)} \right)
&= \frac{Z_{N,\beta}(F,\tilde{f}_N)}{Z_{N,\beta}(F,0)}
= \frac{\mathcal{Z}_{N,\beta}\left(F,\frac{\tilde{f}_N}{\sqrt{\beta}}\right)}{\mathcal{Z}_{N,\beta}(F,0)}
\\ &
= e^{\frac{1}{2 \beta}\left(\tilde{f}_N,\mathcal{C}^q \tilde{f}_N\right)} \frac{Z_N\left(\mathcal{K}_{F,\beta,\mathcal{U}},\mathcal{C}_{\Lambda_N}^q \frac{\tilde{f}_N}{\sqrt{\beta}}\right)}{Z_N(\mathcal{K}_{F,\beta,\mathcal{U}},0)}.
\end{align*}
A standard argument (see Proposition 4.7 in \cite{Hil16} or the proof of Theorem 2.7 in \cite{ABKM}) shows that
$$
\left(
\tilde{f}_N,\mathcal{C}^q_{\Lambda_N} \tilde{f}_N
\right)
\rightarrow 
\left(
f,\mathcal{C}^q_{\T^d} f
\right)_{L^2(\T^d)},
\quad \text{as } N \rightarrow \infty,
$$
and from Theorem \ref{Thm:RepresentationPartitionFunction} it follows that
$$
\left|Z_N(\mathcal{K}_{F,\beta,\mathcal{U}},0)-1 \right|,
 \left|Z_N\left(\mathcal{K}_{F,\beta,\mathcal{U}},\mathcal{C}_{\Lambda_N}^q \frac{f_N}{\sqrt{\beta}}\right)-1 \right|
 \rightarrow 0 \text{ as } N \rightarrow \infty.
$$
This concludes the proof.
\end{proof}

\begin{proof}[Proof of Theorem \ref{Thm:SurfaceTension}]
The proof is similar to the one in \cite{ABKM} but without the need for taking a subsequence. We sketch the main steps here.

Let $\tilde{\zeta}$ be the parameter from Lemma \ref{Lemma:From_K_to_V}, and let $L_0$ and $\epsilon_0$ be as in Theorem \ref{Thm:RepresentationPartitionFunction}. Then, by Lemma \ref{Lemma:From_K_to_V}, there is $\delta_0 > 0$ and $\beta_0 \geq 1$ such that for all $\beta \geq \beta_0$ and $F \in B_{\delta_0}(0)$ we have $\mathcal{K}_{F,\beta,\mathcal{U}} \in B_{\epsilon_0}(0) \subset \mathbf{E}_{\tilde{\zeta},\mathcal{Q}_{\mathcal{U}}}$. Hence we can apply the representation of the partition function in Theorem \ref{Thm:RepresentationPartitionFunction} and we can rewrite the finite-volume surface tension as follows:
\begin{align*}
W_{N,\beta}(F) &= - \frac{1}{\beta L^{Nd}} \ln Z_{N,\beta}(F,0)
\\
&= \mathcal{U}(\bar{F}) - \frac{1}{\beta L^{Nd}} \ln Z_{N,\beta}^{(0)} - \frac{1}{\beta L^{Nd}} \ln \mathcal{Z}_{N,\beta}(F,0)
\\
&= \mathcal{U}(\bar{F})  - \frac{1}{\beta L^{Nd}} \ln Z_{N,\beta}^{(0)}
+ \frac{\lambda(\mathcal{K}_{F,\beta,\mathcal{U}})}{\beta} - \frac{1}{\beta L^{Nd}} \ln \frac{Z_N^{(q(\mathcal{K}_{F,\beta,\mathcal{U}}))}}{Z_N^{(0)}} - \frac{1}{\beta L^{Nd}} \ln Z_N(\mathcal{K}_{F,\beta,\mathcal{U}},0).
\end{align*}

The assumptions \eqref{eq:AssumptionsW} on the potential $\mathcal{U}$ in Theorem \ref{Thm:SurfaceTension} imply that there is $\delta_1 > 0$ such that for $F \in B_{\delta_1}(0)$
$$
\left|
D^2 \mathcal{U} (\bar{F})(z,z) - \mathcal{Q}_{\mathcal{U}}(z)
\right|
= \left| 
D^2 \mathcal{U}(\bar{F})(z,z) - D^2 \mathcal{U}(0)(z,z)
\right|
\leq \frac{\omega_0}{2} |z|^2
$$
and thus
$$
D^2 \mathcal{U}(\bar{F})(z,z) \geq \frac{\omega_0}{2} |z|^2.
$$

The second term $\frac{1}{\beta L^{Nd}} \ln Z_{N,\beta}^{(0)}$ is independent of $F$.

Our next concern is to show that
$$
\mathcal{W}_{N,\beta}(F) = 
\frac{\lambda(\mathcal{K}_{F,\beta,\mathcal{U}})}{\beta} - \frac{1}{\beta L^{Nd}} \ln \frac{Z_N^{(q)}}{Z_N^{(0)}} - \frac{1}{\beta L^{Nd}} \ln Z_N(\mathcal{K}_{F,\beta, \mathcal{U}},0)
$$
is $C^{r_1}$ uniformly in $N$.
The map $F \mapsto \lambda(\mathcal{K}_{F,\beta,\mathcal{U}})$ is $C^{r_1}$ uniformly in $N$ by Theorem~\ref{Thm:RepresentationPartitionFunction} and then chain rule. Similar arguments apply to the second term (see Lemma~11.2 in \cite{ABKM}). The third term is $C^{r_1}$ by smoothness of $Z_N(\mathcal{K})$ in $\mathcal{K}$ with uniform bounds in $N$ as stated in Theorem \ref{Thm:RepresentationPartitionFunction}. Thus there is a constant $\Xi>0$ independent of $\beta$ and $\delta$ such that
$$
\left|
D^2 \mathcal{W}_{N,1}(F)(z,z)
\right|
\leq \Xi |z|^2.
$$

In summary, with the choice $\beta_1 = \frac{4 \Xi}{\omega_0}$ for $\beta \geq \max\lbrace \beta_0,\beta_1 \rbrace$, $\delta \leq \min\lbrace\delta_0,\delta_1\rbrace$ and $F \in B_{\delta}(0)$, we get
\begin{align*}
D^2 W_{N,\beta}(u)(z,z)
&= D^2 \mathcal{U}(\bar{F}) (z,z) + D^2 \mathcal{W}_{N,\beta}(u)(z,z)
\\ &
\geq \frac{\omega}{2} |z|^2 - \frac{\Xi}{\beta}|z|^2
\geq \frac{\omega}{4}|z|^2.
\end{align*}
The uniform convexity of $W_{\beta}(F)$ follows by using the fact that the pointwise limit of uniformly convex functions is uniformly convex.
\end{proof}

\section{Renormalisation group analysis} \label{sec:RG-analysis_BulkFlow}

The proof of Theorem \ref{Thm:RepresentationPartitionFunction} is carried out by renormalisation group analysis. This is an iterative averaging process over different scales. We will introduce the multiscale method in this section and prove Theorem \ref{Thm:RepresentationPartitionFunction}.
We start by motivating the idea of RG.

~\\
We aim to get an expression for
$$
\mathcal{Z}_N(\mathcal{K}, \mathcal{Q},f) 
= \int_{\chi_N} e^{(f,\varphi)} \sum_{X \subset \Lambda_N} \prod_{x \in X} \mathcal{K}(D\varphi(x)) \mu_{\mathcal{Q}}(\de\varphi),
$$
where $\mathcal{Q}$ is a quadratic form, $f \in \chi_N$, $\mathcal{K}\in \mathbf{E}_{\zeta, \mathcal{Q}}$, and $\zeta \in (0,1)$ fixed.
Remember that
$$
\mathcal{C}^q_{\Lambda_N} = \left( \mathcal{A}^q_{\Lambda_N} \right)^{-1},
\quad
\mathcal{A}^q_{\Lambda_N}
= \sum_{i,j=1}^d \left( \mathbf{Q}_{ij} - q_{ij} \right) \nabla_j^* \nabla_i,
$$
is the covariance of the Gaussian free field on $\Lambda_N$. For ease of notation, we will drop the subscript $\Lambda_N$ from now on.

~\\
To sketch the rough idea of the method, set $f=0$ and let us denote
$$
F(\varphi) = \sum_{X \subset \Lambda_N} \prod_{x \in X} \mathcal{K}(D \varphi(x)).
$$
The starting point is to put an additional parameter $q$ into the measure,
\begin{align*}
\mathcal{Z}_{N}(\mathcal{K}, \mathcal{Q},0) 
&= \int_{\chi_N} F(\varphi) \,\mu_{\mathcal{Q}}(\de \varphi) 
= \frac{Z^{(q)}}{Z^{(0)}} \int_{\chi_N} F^q(\varphi) \,\mu_{\mathcal{C}^{q}}(\de \varphi),
\\
\text{where }&
F^q(\varphi) = e^{\frac{1}{2}\sum_{i,j=1}^d
\left(\nabla_i\varphi, q_{ij} \nabla_j \varphi\right)} F(\varphi).
\end{align*}
With the help of the implicit function theorem we "tune" $q$ to find the "correct" Gaussian measure producing a useful formula for the partition function.

~\\
A \textit{finite-range decomposition} of $\mu_{\mathcal{C}^{q}} = \mu_{\mathcal{C}_1} \ast \ldots \ast \mu_{\mathcal{C}_{N}}$ enables us to integrate out iteratively scale by scale,
\begin{align*}
\int_{\chi_N} F^q(\varphi + \phi) \mu_{\mathcal{C}^{q}}(\de \varphi)
&=
\int_{\chi_N} F^q(\xi_1 + \ldots + \xi_{N} + \phi) \mu_{\mathcal{C}_1}(\de\xi_1) \ldots \mu_{\mathcal{C}_{N}}(\de\xi_{N})
\\ &
= \int_{\chi_N} F_1^q(\xi_2 + \ldots + \xi_{N} + \phi) \mu_{\mathcal{C}_2}(\de\xi_2)\ldots \mu_{\mathcal{C}_{N}}(\de\xi_{N})
\\
&= \ldots
\\ &
= \int_{\chi_N} F_{N-1}^q(\xi_{N} + \phi) \mu_{\mathcal{C}_{N}}(\de \xi_{N})
= F_N^q(\phi)
.
\end{align*}

$F^q$ can be written by polymer expansion as,
\begin{align*}
&F^q = \sum_{X \subset \Lambda} e^{-H_0(X)} K_0(\Lambda \setminus X)
= \left(e^{-H_0} \circ K_0\right) (\Lambda),
\\ \text{where} \quad
&H_0(\varphi)(X) = \sum_{x \in X} \sum_{i,j=1}^d \nabla_i\varphi(x) q_{ij} \nabla_j\varphi(x),
\\
\text{and} \quad
&K_0(\varphi)(Y) = e^{- H_0(\varphi)(Y)}
 \prod_{x \in Y} \mathcal{K} \left(D\varphi(x)\right).
\end{align*}
This decomposition can be maintained on each scale $k\in \lbrace 1, \ldots, N \rbrace$, that is there are maps $(H_k^q, K_k^q)$ such that $F_k^q = e^{-H_k^q} \circ K_k^q$. This so-called \textit{circ product} acts on scale~$k$ with polymers consisting of \textit{$k$-blocks}, which are cubes of side length $L^k$ (a precise definition can be found in \eqref{eq:defn_circ_product} in Subsection \ref{subsubsec:polymers_functionals_norms}). At the last scale $N$ there is only one block left, namely the whole set $\Lambda_N$, and the circ product is just a sum of two terms, $\left(e^{-H_N^q} + K_N^q\right)(\Lambda)$.

 The maps $H_k^q$ are the \textit{relevant} (more precisely: relevant and marginal) directions which collect all increasing (and constant) parts in the procedure $F \mapsto \mu_{k+1}\ast F$ and they will live in finite dimensional spaces. The flow $(H,K) \mapsto H_+ = \mathbf{A}^qH + \mathbf{B}^qK$ will be defined in such a way that $(H,K)\mapsto K_+$ is a contraction (by a suitable choice of the map $\mathbf{B}^q$). Moreover, the linear part of $H$ should remain relevant, so that $H$ appears in $K_+$ to second order (by a suitable choice of the map $\mathbf{A}^q$). Then the implicit function theorem can be applied to the flow to find the stable manifold for the initial condition $(H_0,K_0)$ so that the flow converges to its fixed point $(0,0)$.

~\\
This method is described and performed in detail in \cite{BS1}, \cite{BS2}, \cite{BS3}, \cite{BS4} and \cite{BS5} and adapted to gradient models in \cite{AKM16} and \cite{ABKM}.
For the convenience of the reader we review the relevant material from \cite{ABKM} without proofs, see Subsection~\ref{subsec:finite-volume_bulk_single_step}.

For the asserted improvement in Theorem \ref{Thm:RepresentationPartitionFunction}, namely the $N$-independence of the maps $\lambda(\mathcal{K})$ and $q(\mathcal{K})$, we will need some additional properties which we will state explicitly as extensions from \cite{ABKM}. These are the restriction property and $\Z^d$-property as stated in Propositions \ref{Prop:field_locality} and \ref{Prop:Zd-property}, an improved bound on the first derivative of the irrelevant part in Lemma \ref{lem:First_der_improved_large_polymers}, and the single step estimate in Proposition \ref{Prop:SingleStepRG_BulkFlow}.

In Subsection \ref{subsection:Inf_vol_bulk_flow} the flow in \cite{ABKM} will be extended to an infinite-volume flow and the stable manifold theorem will be applied to this flow instead on the finite-volume flow as in \cite{ABKM}.

Finally, estimates on the finite-volume flow and the proof of Theorem \ref{Thm:RepresentationPartitionFunction} will be deduced (see Subsection \ref{subsec:Back_to_finite-volume}).



\subsection{Finite-volume flow and single step estimates} \label{subsec:finite-volume_bulk_single_step}

We start by describing the \textit{finite-range decomposition} of the measure $\mu_{\mathcal{C}^q}$. This decomposition is the starting point for the iterative procedure.

\subsubsection{Finite-range decomposition}\label{sec:FRD}

The operator $\mathcal{A}^{q}: \chi_N \rightarrow \chi_N$ commutes with translations and so does its inverse $\mathcal{C}^{q}$. Thus there exists a unique kernel $C^{q}: \Lambda_N \rightarrow \R$ with $\sum_{x \in \Lambda_N} C^{q}(x) = 0$ such that
$$
\mathcal{C}^{q}\varphi (x) = \sum_{y \in \Lambda_N}C^{q}(x-y) \varphi(y).
$$
The next proposition is Theorem 2.3 in \cite{BUC18}.

\begin{prop}[Finite-range decomposition]\label{Prop:FRD_Buchholz}
Fix $q \in \R^{(d \times m) \times (d \times m)}_{\text{sym}}$ such that $\mathcal{C}^{q}$ is positive definite. Let $L > 3$ be an integer and $N \geq 1$. Then there exist positive, translation invariant operators $\mathcal{C}_k^{q}$ such that
\begin{align*}
\mathcal{C}^{q} &= \sum_{k=1}^{N+1} \mathcal{C}^{q}_k,
\\
C^{q}_k(x) &= -M_k \quad \text{for} \quad |x|_{\infty} \geq \frac{L^k}{2}, \quad k \in \lbrace 1, \ldots, N \rbrace,
\end{align*}
where $M_k \geq 0$ is a constant, positive semi-definite matrix that is independent of $q$. The following bounds hold for any positive integer $l$ and any multiindex $\alpha$:
\begin{align*}
\sup_{x \in \Lambda_N}
\sup_{\Vert \dot{q} \Vert \leq \frac{1}{2}} \Big\vert \nabla^{\alpha} D^l_q C^{q}_k(x)(\dot{q}, \ldots, \dot{q}) \Big\vert
\leq
\begin{cases}
C_{\alpha,l} L^{-(k-1)(d-2+|\alpha|)} & \quad \text{for } d + |\alpha| > 2 \\
C_{\alpha,l} \ln(L) L^{-(k-1)(d-2+|\alpha|)} & \quad \text{for } d + |\alpha| = 2 .
\end{cases}
\end{align*}
Here, $C_{\alpha,l}$ denotes a constant that does not depend on $L$, $N$, and $k$.
\end{prop}
In \cite{BUC18} further bounds in Fourier space are stated. 
For the sake of simplicity they are omitted here.

~\\
In contrast to \cite{ABKM} we combine the last two covariances to a single one:
\begin{align}
\mathcal{C}_{N,N}^{q} = \mathcal{C}_N^{q} + \mathcal{C}_{N+1}^{q}. \label{eq:Last_Scale_Cov}
\end{align}
We will use the following decomposition:
\begin{align}
\mathcal{C}^{q} &= \sum_{k=1}^{N-1} \mathcal{C}^{q}_k + \mathcal{C}^{q}_{N,N},
\label{eq:dec}
\end{align}
where the last term is different from \cite{ABKM}.
The reason for this change is that we extend the \cite{ABKM} flow to infinite volume. In order to have good estimates for the finite-volume covariance we have to perform the last step of integration in the RG flow instead of dealing with a remaining integral in $\int e^{-H_N} + K_N \, \de\mu_{N+1}$ at the last step.

Let us denote by $\mu_k$ the Gaussian measure with covariance $\mathcal{C}_k^q$.

For the sake of completeness we state the following property of Gaussian measures. A~proof can be found, e.g., in \cite{Bry09}.

\begin{lemma}\label{lemma:Dec_Gaussian_measure}
Let $C_k$ be a family of positive definite operators such that $C=\sum_{k} C_k$. Then a field $\varphi$ which is distributed according to $\mu_C$ can be written as $\varphi = \sum_{k} \xi_k$ where $\xi_k$ is distributed according to $\mu_{C_k}$.
\end{lemma}

Another property of the finite range decomposition is independence of $N$, which is stated in Remark 2.4 in \cite{BUC18}. We need this property in order to expand the flow in \cite{ABKM} to infinite volume.

\begin{remark}[Independence of $N$]\label{Rem:FRD_Indep_of_N}
Let $N < N'$ and $\Lambda_N \subset \Lambda_{N'}$ be the corresponding tori. Let us denote by $C_k^N$ and $C_k^{N'}$ the kernels of the decomposition depending on the torus size $L^N$, and $M_k^N$, $M_k^{N'}$ be the corresponding constants from Proposition \ref{Prop:FRD_Buchholz}. It can be shown that for $k < N \leq N'$ and $x \in \Lambda_N$ the decomposition satisfies
\begin{align}
C_k^N (x) - C_k^{N'}(x) = -\left(M_k^N - M_k^{N'}\right), \label{eq:FRD_dep_on_N}
\end{align}
hence the kernels agree up to a constant shift locally, and they are constant for $|x|_{\infty} \geq L^k/2$. We define $\Lambda_N' = \lbrace x \in \Z^d: |x|_{\infty} < (L^N-1)/4\rbrace$. Then we have $x-y \in \Lambda_N$ for $x,y \in \Lambda_N'$. Let $x,y \in \Lambda_N'$ such that $x+e_i,y+e_j \in \Lambda_N'$. Then \eqref{eq:FRD_dep_on_N} implies that
\begin{align*}
\mathbb{E}_{\mu_k^N} \nabla_i \varphi(x) \nabla_j \varphi(y)
&= \nabla_j^* \nabla_i C_k^N(x-y)
= \nabla_j^* \nabla_i C_k^{N'}(x-y)
\\
&= \mathbb{E}_{\mu_k^{N'}} \nabla_i \varphi(x) \nabla_j \varphi(y).
\end{align*}
This means that the covariance structures of $\mu_k^N$ and $\mu_k^{N'}$ agree locally. In particular we can conclude that for any set $X \subset \Lambda_N'$ satisfying $X + e_i \subset \Lambda_N'$ for $1 \leq i \leq d$, any $1 \leq k \leq N$, and any measurable functional $F:\R^X \rightarrow \R$
\begin{align*}
\int_{\chi_N} F(\nabla \varphi \vert_X) \mu_k^N(\de\varphi)
=
\int_{\chi_{N'}} F(\nabla \varphi \vert_X) \mu_k^{N'}(\de\varphi).
\end{align*}
\end{remark}

\subsubsection{Polymers, functionals and norms}\label{subsubsec:polymers_functionals_norms}

As mentioned in the preface to Section \ref{sec:RG-analysis_BulkFlow}, we apply an iterative averaging process over various scales. In this subsection, we discuss several key notions and introduce the setting of the scales and spaces for functionals. We follow closely the presentation in \cite{ABKM}.

~\\
Fix $R = \max \lbrace R_0, 2 \lfloor d/2 \rfloor + 3 \rbrace$.
At each scale $k$ we pave the torus with blocks of side length $L^k$. These so-called \textit{$k$-blocks} are translations by $(L^k \Z)^d$ of the block $B_0 = \left\lbrace z \in \Z^d: |z_i| \leq \frac{L^k-1}{2} \right\rbrace $. Together, they form the set of $k$-blocks denoted by
$$
\mathcal{B}_k = \lbrace B: B \text{ is a $k$-block} \rbrace.
$$
Unions of blocks are called \textit{polymers}. For $X \subset \Lambda$ let $\mathcal{P}_k(X)$ be the set of all $k$-polymers in $X$ at scale~$k$.

Furthermore we need the following notations:
\begin{itemize}
\item
A polymer $X$ is \textit{connected} if for any $x,y \in X$ there is a path $x_1 = x,x_2, \ldots$, $x_n = y$ in $X$ such that $|x_{i+1} - x_i|_{\infty}=1$ for $i=1,\ldots, n-1$. The set of all connected $k$-polymers in $X$ is denoted by $\mathcal{P}_k^c(X)$. The set of connected components of a polymer $X$ is denoted by $\mathcal{C}_k(X)$.
\item
Let $\mathcal{B}_k(X)$ be the set of $k$-blocks contained in $X$ and $|X|_k = |\mathcal{B}_k(X)|$ be the number of $k$-blocks in $X$.
\item
The \textit{closure} $\bar{X} \in \mathcal{P}_{k+1}$ of $X \in \mathcal{P}_k$ is the smallest $(k+1)$-polymer containing~$X$.
\item
The set of \textit{small polymers} $\mathcal{S}_k$ is given by all polymers $X \in \mathcal{P}_k^c$ such that $|X|_k \leq 2^d$. The other polymers in $\mathcal{P}_k \setminus \mathcal{S}_k$ are \textit{large}.
\item
For any block $B \in \mathcal{B}_k$ let $\hat{B} \in \mathcal{P}_k$ be the cube of side length $(2^{d+1}+1)L^k$ centered at~$B$.
\item
The \textit{small neighbourhood} $X^{*} \in \mathcal{P}_{k-1}$ of $X \in \mathcal{P}_k$ is defined by
$$
X^* = \bigcup_{B \in \mathcal{B}_{k-1}(X)} \hat{B}.
$$
\item
The \textit{large neighbourhood} $X^+$ of $X \in \mathcal{P}_k$ is defined by
$$
X^+ = \bigcup_{\substack{B \in \mathcal{B}_k:\\ B \text{ touches }X}} B \cup X.
$$
\end{itemize}

Additionally, we introduce a class of functionals.
\begin{itemize}
	\item Let $M(\mathcal{V}_N)$ be the set of measurable real functions on $\mathcal{V}_N$ with respect to the Borel-$\sigma$-algebra.
	\item Let $\mathcal{N}$ be the space of real-valued functions of $\varphi$ which are in $C^{r_0}$.
	\item A map $F: \mathcal{P}_k \rightarrow \mathcal{N}$ is called \textit{translation invariant} if for every $y \in (L^k \Z)^d$ we have $F(\tau_y(X),\tau_y(\varphi)) = F(X,\varphi)$ where $\tau_y(B) = B+y$ and $\tau_y \varphi(x) = \varphi(x-y)$.
	\item A map $F: \mathcal{P}_k \rightarrow \mathcal{N}$ is called \textit{local} if $\varphi \big\vert_{X^*} = \psi \big\vert_{X^*}$ implies $F(X,\varphi) = F(X,\psi)$.
	\item A map $F: \mathcal{P}_k \rightarrow \mathcal{N}$ is called \textit{shift invariant} if $F(X, \varphi + \psi) = F(X,\varphi)$ for $\psi$ such that $\psi(x) = c$, $x \in X^*$ on each connected component of $X^*$.
\end{itemize}

We set
\begin{align*}
M(\mathcal{P}_k, \mathcal{V}_N)
= \lbrace
F: \mathcal{P}_k \rightarrow \mathcal{N} \big\vert
F(X) \in M(\mathcal{V}_N), F \text{ translation inv., shift inv., local}
\rbrace.
\end{align*}

Notice that we included $C^{r_0}$-smoothness in the definition of the space $M(\mathcal{P}_k,\mathcal{V}_N)$ which is not done in \cite{ABKM}.

Generalisations of $M(\mathcal{P}_k,\mathcal{V}_N)$ are given by $M(\mathcal{P}^c_k,\mathcal{V}_N)$, $M(\mathcal{S}_k,\mathcal{V}_N)$ and $M(\mathcal{B}_k,\mathcal{V}_N)$ where the first component is changed appropriately. We will write $M(\mathcal{P}_k)$, $M(\mathcal{P}^c_k)$, $M(\mathcal{S}_k)$ and $M(\mathcal{B}_k)$ for short.

~\\
The \textit{circ product} of two functionals $F,G \in M(\mathcal{P}_k)$ is defined by
\begin{align}
(F \circ G)(X) = \sum_{Y \in \mathcal{P}_k(X)} F(Y) G(X \setminus Y).
\label{eq:defn_circ_product}
\end{align}

The space of \textit{relevant Hamiltonians} $M_0(\mathcal{B}_k)$, a subspace of $M(\mathcal{B}_k)$, is given by all functionals of the form
$$
H(B,\varphi) = \sum_{x \in B} \mathcal{H}\left(\lbrace x \rbrace, \varphi\right)
$$
where $\mathcal{H}(\lbrace x \rbrace, \varphi)$ is a linear combination of the following \textit{relevant monomials}:
\begin{itemize}
\item
The constant monomial $M(\lbrace x \rbrace)_{\es}(\varphi) = 1$;
\item
the linear monomials $M(\lbrace x \rbrace)_{\beta}(\varphi) = \nabla^{\beta} \varphi(x)$ for $1 \leq |\beta| \leq \lfloor \frac{d}{2} \rfloor + 1$;
\item
the quadratic monomials $M(\lbrace x \rbrace)_{\beta,\gamma}(\varphi) = \nabla^{\beta} \varphi(x) \nabla^{\gamma} \varphi(x)$ for $1 = |\beta| = |\gamma|$.
\end{itemize}

~\\
Next we introduce norms on the space of functionals. Fix $r_0 \in \N$, $r_0 \geq 3$.

\begin{itemize}
\item
Define
\begin{align*}
&\bigoplus_{r=0}^{\infty} \mathcal{V}_N^{\otimes r}
\\ & \quad
= \left\lbrace g = \left(g^{(0)}, g^{(1)}, \ldots \right) \Big\vert \, g^{(r)} \in \mathcal{V}_N^{(r)}, \text{ only finitely many non-zero elements} \right\rbrace.
\end{align*}

The space of test function is given by
$$
\Phi = \Phi_{r_0} = \left\lbrace g \in \bigoplus_{r=0}^{\infty} \mathcal{V}_N^{\otimes r}: g^{(r)} = 0 \,\, \forall r \geq r_0 \right\rbrace.
$$

A norm on $\Phi$ is given as follows: On $\mathcal{V}_N^{\otimes 0} = \R$ we take the usual absolute value on~$\R$.
For $\varphi \in \mathcal{V}_N$ we define
\begin{align*}
\vert \varphi \vert_{j,X} = \sup_{x \in X^*} \sup_{1 \leq |\alpha| \leq p_{\Phi}} \mathfrak{w}_j(\alpha)^{-1} \big\vert \nabla^{\alpha} (\varphi)(x) \big\vert
\end{align*}
where $\mathfrak{w}_j(\alpha) = h_j L^{-j|\alpha|} L^{-j \frac{d-2}{2}}$, $h_j = 2^j h$ and $p_{\Phi} = \left\lfloor \frac{d}{2} \right\rfloor + 2$.
For $g^{(r)} \in \mathcal{V}_N^{\otimes r}$ we define
\begin{align*}
&\left\vert g^{(r)} \right\vert_{j,X}
\\
&\quad
= \sup_{x_1, \ldots, x_r \in X^*} \sup_{1 \leq |\alpha_1|, \ldots, |\alpha_r| \leq p_{\Phi}} \left(\prod_{l=1}^r \mathfrak{w}_j(\alpha_l)^{-1}\right) \nabla^{\alpha_1} \otimes \ldots \otimes \nabla^{\alpha_r} g^{(r)}(x_1, \ldots, x_r).
\end{align*}
Then set $|g|_{j,X} = \sup_{r \leq r_0} \left|g^{(r)}\right|_{j,X}$.

\item
A homogeneos polynomial $P^{(r)}$ of degree $r$ on $\mathcal{V}_N$ can be uniquely identified with a symmetric $r$-linear form and hence with an element $\overline{P^{(r)}}$ in the dual of $\mathcal{V}_N^{\otimes r}$.
So we can define the pairing
$$
\langle P,g \rangle = \sum_{r=0}^{\infty} \left\langle \overline{P^{(r)}}, g^{(r)} \right\rangle
$$
and a norm
$$
\vert P \vert_{j,X} = \sup \left\lbrace \langle P,g \rangle: g \in \Phi, |g|_{j,X} \leq 1 \right\rbrace.
$$
For $F \in C^{r_0}(\mathcal{V}_N)=\mathcal{N}^{\es}$ the pairing is given by $\langle F,g \rangle_{\varphi} = \langle \mathrm{Tay}_{\varphi} F,g \rangle$ which defines a norm
$$
\vert F \vert_{j,X,T_{\varphi}} = \vert \mathrm{Tay}_{\varphi} F \vert_{j,X}
= \sup \left\lbrace \langle F,g \rangle_{\varphi}: g \in \Phi, |g|_{j,X} \leq 1 \right\rbrace.
$$
Here, $\mathrm{Tay}_{\varphi} F $ denotes the Taylor polynomial of order $r_0$ of $F$ at $\varphi$.

\item
Let $F \in M(\mathcal{P}_k^c)$. In \cite{ABKM} weights $W_k^{X}, w_k^X, w_{k:k+1}^X \in M(\mathcal{P}_k)$ are defined. Useful properties are summarized in Lemma \ref{lemma:Properties_of_weights} below. Weighted norms are given~by
\begin{align*}
\vertiii{F(X)}_{k,X} &= \sup_{\varphi} \vert F(X) \vert_{k,X,T_{\varphi}} W_k^X(\varphi)^{-1},
\\
\Vert F(X) \Vert_{k,X} &= \sup_{\varphi} \vert F(X) \vert_{k,X,T_{\varphi}} w_k^X(\varphi)^{-1},
\\
\Vert F(X) \Vert_{k:k+1,X} &= \sup_{\varphi} \vert F(X) \vert_{k,X,T_{\varphi}} w_{k:k+1}^X(\varphi)^{-1}.
\end{align*}

\item
The global weak norm for $F \in M(\mathcal{P}_k^c)$ for $A \geq 1$ is given by
$$
\Vert F \Vert_k^{(A)} = \sup_{X \in \mathcal{P}_k^c} \Vert F(X) \Vert_{k,X} A^{|X|_k}.
$$

\item A norm on relevant Hamiltonians is given as follows. For $H \in M_0(\mathcal{B}_k)$ we can write
$$
H(B,\varphi) = \sum_{x \in B}
\left(
 a_{\es} +
 \sum_{\beta \in \mathfrak{v}_1} a_{\beta} \nabla^{\beta} \varphi(x) +
\sum_{x \in B} \sum_{\beta,\gamma \in \mathfrak{v}_2} a_{\beta,\gamma} \nabla^{\beta} \varphi(x) \nabla^{\gamma} \varphi(x)
\right).
$$
Here
\begin{align*}
\mathfrak{v}_1 &= \left\lbrace \beta \in \N_0^{\mathcal{U}}, 1 \leq |\beta| \leq \left\lfloor \frac{d}{2} \right\rfloor + 1 \right\rbrace,
\\
\mathfrak{v}_2 &= \left\lbrace (\beta, \gamma) \in \N_0^{\mathcal{U}} \times \N_0^{\mathcal{U}}, |\beta| = |\gamma|=1, \beta < \gamma \right\rbrace,
\end{align*}
where $\mathcal{U} = \lbrace e_1, \ldots, e_d \rbrace$ and the expression $\beta < \gamma$ refers to any ordering of $\lbrace e_1, \ldots, e_d \rbrace$.
With these preparations we define a norm on $M_0(\mathcal{B}_k)$ as follows:
$$
\Vert H \Vert_{k,0} =
L^{dk} \left| a_{\es}\right|
+ \sum_{\beta \in \mathfrak{v}_1} h_k L^{kd} L^{-k\frac{d-2}{2}} L^{-k|\beta|} \left| a_{\beta} \right|
+ \sum_{(\beta,\gamma)\in \mathfrak{v}_2} h_k^2 \left| a_{(\beta,\gamma)} \right|.
$$

\end{itemize}

For the sake of completeness we review Proposition 7.1 from \cite{ABKM}.
The last scale weights ($k=N$) differ from \cite{ABKM} due to the modified definition of the last scale covariance (see \eqref{eq:Last_Scale_Cov}).
However, this does not change the properties of the weights as stated in the following lemma.

\begin{lemma}
\label{lemma:Properties_of_weights}
Let $L \geq 2^{d+3} + 16 R$. The weight functions $w_k$, $w_{k:k+1}$ and $W_k$ are well-defined and satisfy the following properties:
\begin{enumerate}
\item
For any $Y \subset X \in \mathcal{P}_k$, $0 \leq k \leq N$, and $\varphi \in \mathcal{V}_N$
\begin{align*}
w_k^Y(\varphi) \leq w_k^X(\varphi)
\quad\mathrm{and}\quad
w_{k:k+1}^Y(\varphi) \leq w_{k:k+1}^X(\varphi).
\end{align*}
\item
For any strictly disjoint polymers $X,Y \in \mathcal{P}_k$, $0 \leq k \leq N$, and $\varphi \in \mathcal{V}_N$
\begin{align*}
w_{k}^{X \cup Y}(\varphi) = w_k^X(\varphi) w_k^Y(\varphi).
\end{align*}
\item
For any polymers $X,Y \in \mathcal{P}_k$ such that $\mathrm{dist}(X,Y) \geq \frac{3}{4} L^{k+1}$, $0 \leq k \leq N$, and $\varphi \in \mathcal{V}_N$
\begin{align*}
w_{k:k+1}^{X \cup Y}(\varphi) = w_{k:k+1}^X(\varphi)w_{k:k+1}^Y(\varphi).
\end{align*}
\item
For any disjoint polymers $X,Y \in \mathcal{P}_k$, $0 \leq k \leq N$, and $\varphi \in \mathcal{V}_N$
\begin{align*}
W_k^{X \cup Y}(\varphi) = W_k^X(\varphi) W_k^Y(\varphi).
\end{align*}
\end{enumerate}
Moreover, there is a constant $h_0 = h_0(L,\zeta)$ such that for all $h \geq h_0$ the weight functions satisfy the following properties:
\begin{enumerate}
\setcounter{enumi}{4}
\item
For any disjoint polymers $X,Y \in \mathcal{P}_k$ and $U = \pi(X) \in \mathcal{P}_{k+1}$, $0 \leq k \leq N-1$, and $\varphi \in \mathcal{V}_N$
\begin{align*}
w_{k+1}^U(\varphi) \geq w_{k:k+1}^X(\varphi) \left( W_k^{U^+}(\varphi)\right)^2.
\end{align*}
\item
For all $0 \leq k \leq N-1$, $X \in \mathcal{P}_{k+1}$ and $\varphi \in \mathcal{V}_N$,
\begin{align*}
e^{\frac{|\varphi|^2_{k+1,X}}{2}} w_{k:k+1}^X(\varphi)
\leq w_{k+1}^X(\varphi).
\end{align*}
\end{enumerate}
Lastly, there exists a constant $\kappa = \kappa(L,\zeta)$ with the following properties:
\begin{enumerate}
\setcounter{enumi}{6}
\item
There is a constant $A_{\mathcal{P}}$ such that for $q \in B_{\kappa}(0)$, $\rho = (1+ \frac{\zeta}{4})^{1/3}-1$, $Y \in \mathcal{P}_k$, $0 \leq k \leq N$, and $\varphi \in \mathcal{V}_N$
\begin{align*}
\left( \int_{\chi_N} \left( w_k^X(\varphi + \xi) \right)^{1+\rho} \mu_{k+1}(\de \xi) \right)^{\frac{1}{1+\rho}}
\leq \left( \frac{A_{\mathcal{P}}}{2} \right)^{|X|_k} w_{k:k+1}^X(\varphi).
\end{align*}
\item
There is a constant $A_{\mathcal{B}}$ independent of $L$ such that for $q \in B_{\kappa}$, $\rho = (1+\frac{\zeta}{4})^{1/3}-1$, $B \in \mathcal{B}_k$, $0 \leq k \leq N$, and $\varphi \in \mathcal{V}_N$
\begin{align*}
\left( \int_{\chi_N} \left( w_k^B(\varphi + \xi) \right)^{1+\rho} \mu_{k+1}(\de \xi) \right)^{\frac{1}{1+\rho}}
\leq \frac{A_{\mathcal{B}}}{2}  w_{k:k+1}^B(\varphi).
\end{align*}
\end{enumerate} 
\end{lemma}

\begin{remark}\label{Remark:Weight_Independent_N}
The weights $w_k$ and $w_{k:k+1}$ are independent of the size of the torus as long as the input polymer is small enough. This can be seen when examining the construction of the weights. The weights essentially arise as follows: Take the local quadratic form from [AKM16], integrate against the covariance of the finite-range decomposition (which is independent of the size of the torus by Remark \ref{Rem:FRD_Indep_of_N}) and add explicit local perturbing terms. These steps are independent of the size of the torus as long as the input-polymer is small enough compared to the torus.

The weights $W_k$ are given explicitly and obviously local.
\end{remark}

Aside from the parameter $L$ two parameters appear above in the definition of the norms: $h$ and $A$.

The parameter $h$ is determined by desired properties for the weights $W_k,w_k,w_{k:k+1}$, see Proposition 7.1 in \cite{ABKM} (cited here in Lemma \ref{lemma:Properties_of_weights}), dependent on the choice of $L$. We will use the weights without explaining the construction and thus we will always choose $h$ large enough as required, depending on $L$.

The parameter $A$ (also dependent on $L$) will be fixed in Proposition \ref{Prop:Existence_BulkFlow}. It will be chosen larger than in \cite{ABKM}.

Finally, there will be a small parameter $\kappa = \kappa(L)$. It constrains the parameter $q \in \R^{(d \times m) \times (d \times m)}_{\text{sym}}$ which determines the Gaussian covariance $\mathcal{C}^q$. The constraint will be that $q \in B_{\kappa}(0)$ for $\kappa$ small. The parameter $\kappa$ is determined by desired properties for the weights, $W_k,w_k,w_{k:k+1}$, see Proposition 7.1 in \cite{ABKM} (cited here in Lemma \ref{lemma:Properties_of_weights}).

\subsubsection{Definition of the renormalisation map}\label{subsubsec:The_renormalisation_map}

We use the finite-range decomposition of $\mathcal{C}^{q}$ into covariances $\mathcal{C}^q_1,\ldots, \mathcal{C}^q_{N-1},\mathcal{C}^q_{N,N}$ defined in Subsection \ref{sec:FRD} (see \eqref{eq:dec}). The decomposition implies that a field $\varphi$ distributed according to $\mu_{\mathcal{C}^{(q)}}$ can be decomposed into fields $\xi_k$ distributed according to $\mu_{\mathcal{C}^q_{k}} =: \mu^q_{k}$,
$$
\varphi \stackrel{\mathcal{D}}{=} \sum_{k=1}^{N} \xi_k,
$$
and that $\mu_{\mathcal{C}^{(q)}} = \mu_1^q \ast \cdots \ast \mu_{N-1}^q \ast \mu^q_{N,N}$ (see Lemma \ref{lemma:Dec_Gaussian_measure}).

Let us define the renormalisation map
\begin{align*}
\mathcal{R}_k F(\varphi) = \int_{\chi_N}F(\varphi + \xi) \mu_k(\de \xi).
\end{align*}
Then
$$
\int_{\chi_N} F(\varphi) \mu_{\mathcal{C}^{(q)}}(\de \varphi)
= \mathcal{R}_{N,N} \mathcal{R}_{N-1} \ldots \mathcal{R}_1(F)(0).
$$

The flow under $\mathcal{R}_k$ will be described by two sequences of functionals $H_k \in M_0(\mathcal{B}_k)$ and $K_k \in M(\mathcal{P}_k^c)$. In the following we define those sequences and state properties as far as it is needed for our purpose of proving Theorem \ref{Thm:RepresentationPartitionFunction} and for the understanding of the extension to observables.

~\\
The flow is given by
\begin{align*}
\mathbf{T}_k: M_0(\mathcal{B}_k) \times M(\mathcal{P}_k^c) \times \R_{\text{sym}}^{d \times d}
&\quad\rightarrow\quad
 M_0(\mathcal{B}_{k+1}) \times M(\mathcal{P}_{k+1}^c),
\\
(H,K,q) &\quad\mapsto\quad (H_+,K_+).
\end{align*}
Note that we sometimes omit the scale $k$ from the notation; if doing so, the $+$ indicates the change of scale from $k$ to $k+1$. The maps $H_+ \in M_0(\mathcal{B}_{k+1})$ and $K_+ \in M(\mathcal{P}_{k+1})$ are chosen such that
$$
\mathcal{R}_+ (e^{-H} \circ K)(\Lambda_N) = (e^{-H_+} \circ K_+)(\Lambda_N).
$$

Let us introduce a projection $\Pi_2: M(\mathcal{B}_k) \rightarrow M_0(\mathcal{B}_k)$ on the space of relevant Hamiltonians. For $F \in M(\mathcal{B}_{k})$, $\Pi_2 F$ is attained as homogenisation of the second order Taylor expansion of $F(B)$ given by $\dot{\varphi} \mapsto  F(B,0) + DF(B,0) \dot{\varphi} + \frac{1}{2} D^2F(B,0)(\dot{\varphi},\dot{\varphi})$.
More precisely, $\Pi_2 F$ is the relevant Hamiltonian  $F(B,0) + l(\dot{\varphi}) + Q(\dot{\varphi},\dot{\varphi})$ where $l$ is the unique linear relevant Hamiltonian that satisfies $l(\dot{\varphi}) = DF(B,0) \dot{\varphi}$ for all $\dot{\varphi}$ who are polynomials of order $\left\lfloor \frac{d}{2} + 1 \right\rfloor$ on $B^+$, and $Q$ is the unique quadratic relevant Hamiltonian that agrees with $\frac{1}{2} D^2 F(B,0)(\dot{\varphi},\dot{\varphi})$ on all $\dot{\varphi}$ which are affine on $B^+$. These heuristics are made precise in \cite{ABKM}, Section 8.4.

~\\
The relevant part of the flow on the next scale, the map $H_+$, is defined as follows: For $B_+ \in  \mathcal{B}_{k+1}$
\begin{align*}
H_+(B_+)
&= \mathbf{A}^q_k H (B_+) + \mathbf{B}^q_k K(B_+)
\\
&= \sum_{B \in \mathcal{B}_{k+1}(B_+)} \Pi_2 \mathcal{R}_{k+1} H(B)
- \sum_{B \in \mathcal{B}_{k+1}(B_+)} \Pi_2 \mathcal{R}_{k+1} K(B).
\end{align*}

\begin{remark}
We comment again on the motivation for the decomposition into $H$ and $K$ (see also at the beginning of this section). $\mathbf{A}^q_k H$ is a linear order perturbation which results in the fact that $H$ appears to second order in $K_+$, see Proposition \ref{Prop:SingleStepRG_BulkFlow}. Moreover, $\mathbf{B}^q_k K$ is defined in such a way that $(H,K)\mapsto K_+$ is a contraction, see Proposition~\ref{Prop:FirstDerivative_BulkFlow}.
\end{remark}

For the definition of the irrelevant part $K_+$ of the flow at the next scale, set
$$
\widetilde{H}(B) = \Pi_2 \mathcal{R}_{k+1} H(B) - \Pi_2 \mathcal{R}_{k+1} K(B),
$$
and for $X \in \mathcal{P}_k$ and $U \in \mathcal{P}_{k+1}$,
\begin{align*}
&\chi(X,U) = \1_{\pi(x)=U}, \quad\mathrm{where}
\\
&\pi(X) = \bigcup_{Y \in \mathcal{C}(X)} \tilde{\pi}(Y) \quad \mathrm{and}
\\
&\tilde{\pi}(Y)
= \begin{cases}
\bar{X} & \text{ if } X \in \mathcal{P}^c \setminus \mathcal{S},
\\
B_+ & \text{ where } B_+ \in \mathcal{B}_+ \text{ with } B_+ \cap X \neq \es \text{ for } X \in \mathcal{S} \setminus \lbrace \es \rbrace,
\\
\es & \text{ if } X = \es.
\end{cases}
\end{align*}
Then
\begin{align}
K_+(U,\varphi) &= \mathbf{S}^q_k (H_+,K_+)(U,\varphi)
\nonumber
\\
&= \sum_{X \in \mathcal{P}} \chi(X,U) \left( e^{\widetilde{H}(\varphi)} \right)^{U \setminus X} \left( e^{\widetilde{H}(\varphi)} \right)^{- X \setminus U}
\nonumber
\\ & \qquad \times
\int 
\left[
\left( 1 - e^{\widetilde{H}(\varphi)} \right)
\circ \left( e^{-H(\varphi + \xi)} - 1 \right)
\circ K(\varphi + \xi)
\right]
(X) \mu_+(\de\xi).
\label{eq:defn_K+}
\end{align}
If the dependence of $\mathbf{S}^q_k$ on $q$ is not of direct importance we omit it from the notation.

~\\
We review the following properties of the map $(H,K) \mapsto K_+$ from Lemma 6.4 in \cite{ABKM}.

\begin{lemma}
Let $L \geq 2^{d+2} + 4R$.
For $H \in M_0(\mathcal{B}_k)$ the functional $K_+$ defined above has the following properties.
\begin{enumerate}
	\item If $K \in M(\mathcal{P}_k)$, then $K_+ \in M(\mathcal{P}_+)$.
	\item If $K \in M(\mathcal{P}_k)$ factors on scale $k$, then $K_+$ factors on scale $k+1$.
\end{enumerate}
\end{lemma}

The construction on $K_+$ gives it a local dependence on $K$, as formulated in the next proposition.

\begin{prop}\label{Prop:field_locality}
The map $(H,K) \mapsto K_+$ satisfies the \normalfont{restriction property}\textit{, that is
for $U \in \mathcal{P}_{k+1}$ the value of $K_+(U)$ depends on $U$ only via the restriction $K \big\vert_{U^*}$ of $K$ to polymers in $\mathcal{P}(U^*)$.}
\end{prop}

\begin{proof}
This follows from the definition of $K_+$ and from the fact that $\mathcal{R}_{k+1}$ preserves locality.
\end{proof}


For the construction of the infinite-volume flow later we  consider the family $(K^{\Lambda})_{\Lambda}$ in dependence on the torus $\Lambda$. More precisely, we consider tori $\Lambda_N$ with increasing side length $L^N$, $N \in \N$. Let $\mathcal{P}_k(\Z^d)$ be the set of finite unions of $k$-blocks in~$\Z^d$. We need the following compatibility condition.

\begin{defn}\label{Defn:Zd-property}
We say that a family of maps $(K^{\Lambda})_{\Lambda}$ satisfies the $(\Z^d)$-property if for any $X \in \mathcal{P}_k(\Z^d)$ and for $\Lambda \subset \Lambda'$ satisfying $\diam (X) \leq \frac{1}{2} \diam(\Lambda)$ it holds that
$$
K^{\Lambda}(X) = K^{\Lambda'}(X).
$$
\end{defn}

Given $(H,K^{\Lambda})$, we note the dependence on $\Lambda$ also in the map $\mathbf{S}^{\Lambda}$. By the definition of the map $(H,K) \mapsto \mathbf{S}^{\Lambda}_k(H,K)$ we directly get the following property.

\begin{prop}\label{Prop:Zd-property}
Let $(K^{\Lambda})_{\Lambda}$ satisfy the $(\Z^d)$-property and let $H \in M_0(\mathcal{B})$. Then $(\mathbf{S}^{\Lambda}(H,K,q))_{\Lambda}$ also satisfies the $(\Z^d)$-property.
\end{prop}

\begin{proof}
Let $U \in \mathcal{P}_{+}(\Z^d)$ such that $\diam (U) \leq \frac{1}{2} \diam (\Lambda)$. Let $\Lambda'$ be a torus larger than $\Lambda$. Then
$$
K^{\Lambda'}_+(U) = \mathbf{S}(H, K^{\Lambda'})(U).
$$
We use the restriction property in Proposition \ref{Prop:field_locality} to see that $\mathbf{S}(H,K^{\Lambda'})(U)$ only depends on $K^{\Lambda'}$ through $K^{\Lambda'}\big\vert_{U^*}$. In fact, no polymers that are larger than $U$ can appear in the formula for $\mathbf{S}_k$ due to the definition of $\chi(X,U)$. Thus for any $X \in \mathcal{P}(U^*)$ that appears in $\mathbf{S}$ it holds that $\diam (X) \leq \frac{1}{2} \diam (\Lambda)$, and we can apply the assumption that $(K^{\Lambda})_{\Lambda}$ satisfies the $(\Z^d)$-property.
\end{proof}

\subsubsection{Properties of the renormalisation map}

Here we state important properties of the renormalisation map $\mathbf{T}_k$, namely smoothness of the irrelevant part (Proposition \ref{Prop:Smoothness_BulkFlow}), an improved bound on the first derivative of the irrelevant part (Lemma \ref{lem:First_der_improved_large_polymers}),
 contractivity of the linearisation of the irrelevant part (Proposition \ref{Prop:FirstDerivative_BulkFlow}), and a single step estimate (Proposition \ref{Prop:SingleStepRG_BulkFlow}). Smoothness and contractivity are proven in \cite{ABKM}, but we add restriction and $(\Z^d)$-property in the statements which will be useful to perform the extension to infinite volume in the next section.

We explicitly analyse the dependence of $\mathbf{S}_k$ on $q$ in the next statement, so we consider $\mathbf{S}_k$ as a map from $M_0(\mathcal{B}_k) \times M(\mathcal{P}_k^c) \times \R^{(d \times m) \times (d \times m)}_{\text{sym}}$ to $M(\mathcal{P}_{k+1})$. This proposition is a small extension of Theorem 6.7 in {\cite{ABKM}}.

\begin{prop}[Smoothness of the bulk flow]\label{Prop:Smoothness_BulkFlow}
Let
$$
U_{\rho,\kappa} = \lbrace
(H,K,q) \in M_0(\mathcal{B}_k) \times M(\mathcal{P}_k^c) \times \R^{(d \times m) \times (d \times m)}_{\text{sym}}:
\Vert H \Vert_{k,0} < \rho, \Vert K \Vert_k^{(A)} < \rho, \Vert q\Vert < \kappa
\rbrace.
$$
There is $L_0$ such that for all integers $L \geq L_0$ there are $A_0$, $h_0$ and $\kappa$ with the following property. For all $A \geq A_0$ and $h \geq h_0$ there exists $\rho = \rho(A)$ such that for all $k \leq N$
$$
\mathbf{S}_k \in C^{\infty}\left( U_{\rho,\kappa}, M(\mathcal{P}_{k+1}^c) \right).
$$
For any $j_1,j_2,j_3 \in \N$ there are constants $C_{j_1,j_2,j_3}$ independent of $N$ such that for any $(H,K,q) \in U_{\rho,\kappa}$
\begin{align*}
\left\Vert
D_1^{j_1} D_2^{j_2} D_3^{j_3} \mathbf{S}_k(H,K,q) (\dot{H}^{j_1}, \dot{K}^{j_2}, \dot{q}^{j_3})
\right\Vert_{k+1}^{(A)}
\leq
C_{j_1,j_2,j_3} \Vert \dot{H} \Vert_{k,0}^{j_1} \left( \Vert \dot{K} \Vert_k^{(A)} \right)^{j_2} \Vert \dot{q} \Vert^{j_3}.
\end{align*}
Moreover, $\mathbf{S}_k(H,K,q)(U)$ satisfies the restriction property and preserves the $(\Z^d)$-property.
\end{prop}

\begin{proof}
The restriction property is stated in Proposition \ref{Prop:field_locality}. The $(\Z^d)$-property is preserved by Proposition \ref{Prop:Zd-property}. The smoothness and bounds are part of Theorem 6.7 in \cite{ABKM}.
\end{proof}

For the transfer of smoothness properties from the global flow back to the finite-volume flow in Proposition \ref{Prop:Existence_BulkFlow} we need the following improved bound on the first derivative of $\mathbf{S}_k$ on long polymers.

\begin{lemma}\label{lem:First_der_improved_large_polymers}
Assume that Proposition \ref{Prop:Smoothness_BulkFlow} holds. Let $\mathcal{P}_{k+1}^2(\Lambda)$ be the set of polymers $U \in \mathcal{P}_{k+1}(\Lambda)$ such that $\diam(U) > \frac{1}{2} \diam(\Lambda)$. Then, for any $x \in (0,2\alpha)$, where $\alpha = \left[ (1+2^d)(1+6^d) \right]^{-1}$, and for any $(H,K) \in U_{\rho}$,
$$
\left\Vert
D_H D_K D_q \mathbf{S}_k (H,K,q)(\dot{H},\dot{K},\dot{q}) \big\vert_{\mathcal{P}_{k+1}^2(\Lambda)}
\right\Vert_{k+1}^{(A)}
\leq C_1
 A^{-\frac{x}{2}L^{N-(k+1)}} A^4
\Vert \dot{H} \Vert_{k,0}
\Vert \dot{K} \Vert_{k}^{(A)}
\Vert \dot{q} \Vert.
$$
\end{lemma}

\begin{proof}

When inspecting the proof of Lemma 9.6 in \cite{ABKM}, we get
\begin{align*}
A^{|U|_{k+1}}
&
\left\Vert
D_{I_1} D_{I_2} D_{J} D_{K}
P_1 (I_1,I_2,J,K)(U)(\dot{I_1},\dot{I_2},\dot{J}, \dot{K})
\right\Vert
\\ &
\leq
A^{-x|U|_{k+1}}A^2 
\vertiii{ \dot{I}_1 }
\vertiii{ \dot{I}_2 }
\vertiii{ \dot{J} }
\left\Vert \dot{K} \right\Vert_{k:k+1}^{(A/(2A_{\mathcal{P}}),B) }
\end{align*}
for $x \in (0,2\alpha)$. Namely, we have that
\begin{align*}
A^{|U|_{k+1}}
&
\left\Vert
D_{I_1} D_{I_2} D_{J} D_{K}
P_1 (I_1,I_2,J,K)(U)(\dot{I_1},D_{I_2},D_{J}, D_{K})
\right\Vert
\\ &
\leq
\left(
\frac{(48 A_{\mathcal{P}})^{2L^d}}{A^{2 \alpha}} 
\right)^{|U|_{k+1}}A^2 
\vertiii{ \dot{I}_1 }
\vertiii{ \dot{I}_2 }
\vertiii{ \dot{J} }
\left\Vert \dot{K} \right\Vert_{k:k+1}^{(A/(2A_{\mathcal{P}}),B) }
\\ &
\leq A^{-x|U|_{k+1}}A^2 
\vertiii{ \dot{I}_1 }
\vertiii{ \dot{I}_2 }
\vertiii{ \dot{J} }
\left\Vert \dot{K} \right\Vert_{k:k+1}^{(A/(2A_{\mathcal{P}}),B) }
\end{align*}
if we choose
$$
A \geq \left(48 A_{\mathcal{P}} \right)^{\frac{2 L^d}{2 \alpha - x}}.
$$

This estimate and chain rule implies that in the case of the bulk flow there is a constant $C_1$ such that for any $x \in (0,2 \alpha)$ and $(H,K) \in U_{\rho}$
\begin{align*}
&
A^{|U|_{k+1}}
\left\Vert
D_H D_K D_q \mathbf{S}_k (H,K,q) (\dot{H},\dot{K},\dot{q})(U)
\right\Vert_{k+1,U}
\\ &
\leq
 C_1 A^{-x |U|_{k+1}} A^4
\Vert \dot{H} \Vert_{k,0}
\Vert \dot{K} \Vert_k^{(A)}
\Vert \dot{q} \Vert,
\end{align*}
where the factors $A$ come from the estimates on $D_J P_1$, $DP_3$, and $D_I P_2$.

Thus
\begin{align*}
A^{|U|_{k+1}}
\left\Vert
D_H D_K D_q \mathbf{S}_k (H,K,q) (\dot{H},\dot{K}\dot{q})(U)
\right\Vert_{k+1,U}
\leq C_1 
A^{-x |U|_{k+1}} A^4
\Vert \dot{H} \Vert_{k,0}
\Vert \dot{K} \Vert_k^{(A)}
\Vert \dot{q} \Vert
.
\end{align*}
Since 
$$
\diam(\Lambda) = \sqrt{2}L^N
\quad \text{and} \quad
\diam(U) \leq |U|_k \sqrt{2}L^{k+1}
$$
we get for $U \in \mathcal{P}_{k+1}^{c,2}(\Lambda)$
$$
|U|_{k+1} > \frac{1}{2} L^{N-(k+1)}.
$$
Thus the claim follows.

\end{proof}

The following proposition is Theorem 6.8 in {\cite{ABKM}}.

\begin{prop}[Contractivity of the bulk flow]\label{Prop:FirstDerivative_BulkFlow}
The first derivative of $\mathbf{T}_k$ at $H=0$ and $K=0$ has the triangular form
\begin{align*}
D \mathbf{T}_k(0,0,q)  \begin{pmatrix} \dot
H \\ \dot{K} \end{pmatrix}
= \begin{pmatrix}
	\mathbf{A}^{q}_k & \mathbf{B}^{q}_k \\ 0 & \mathbf{C}^{q}_k
\end{pmatrix}
\begin{pmatrix} \dot
H \\ \dot{K} \end{pmatrix}
\end{align*}
where
\begin{align*}
\mathbf{C}^{q}_k \dot{K}(U) = \sum_{B: \bar{B}= U}(1-\Pi_2) \mathcal{R}_+ \dot{K}(B) + \sum_{\substack{X \in \mathcal{P}_k^c \setminus \mathcal{B}(X)\\ \pi(X) = U}} \mathcal{R}_+ \dot{K}(X).
\end{align*}
For any fixed $\theta \in (0,1)$ there is $L_0$ such that for all integers $L \geq L_0$ there exist $A_0$, $h_0$ and $\kappa$ with the following property. For all $A \geq A_0$, $h \geq h_0$  and for $\Vert q \Vert < \kappa$ the following bounds hold independent of $k$ and $N$:
\begin{align*}
\Vert \mathbf{C}_k^{q} \Vert \leq \theta,
\quad
\Vert \left(\mathbf{A}_k^{q}\right)^{-1} \Vert \leq \frac{3}{4},
\quad
\Vert \mathbf{B}_k^{q} \Vert \leq \frac{1}{3}.
\end{align*}
Moreover, the derivatives of the operators with respect to $q$ are bounded.
\end{prop}

\begin{remark}
In \cite{ABKM} $\theta$ is fixed to be $\frac{3}{4} \eta$. For the single step estimates in Proposition~\ref{Prop:SingleStepRG_BulkFlow} we have to choose $\theta$ smaller than $\frac{1}{2}$. Thus we formulated the Proposition with this additional flexibility. Inspection of the proof of the bound on $\Vert \mathbf{C}^q_k \Vert$ in \cite{ABKM} shows that a smaller $\theta$ can be obtained by choosing larger $L_0$ and~$A_0$.
\end{remark}

~\\
Proposition \ref{Prop:Smoothness_BulkFlow} and Proposition \ref{Prop:FirstDerivative_BulkFlow} can be combined to prove a single step estimate of the irrelevant part of the flow. This bound is not proven in \cite{ABKM}. The estimate will help us deduce estimates on the finite-volume flow given the infinite-volume flow, see Proposition~\ref{Prop:Existence_BulkFlow}.

~\\
For fixed $\eta$ and $\rho_0$, let us introduce the space
\begin{align}
&\mathbb{D}_k(\rho_0,\eta,\Lambda)
\nonumber
\\ & \quad
= 
\left\lbrace (H,K) \in M_0(\mathcal{B}_k) \times M(\mathcal{P}_k(\Lambda)):
H \in B_{\rho_0 \eta^k}(0),
K \in B_{\rho_0 \eta^{2k}}(0)
\right\rbrace.
\label{eq:defn_D_bulk}
\end{align}

\begin{prop}[Single step estimate for the bulk flow]\label{Prop:SingleStepRG_BulkFlow}
Fix $\eta \in (0,1)$. There is $L_0$ such that for all integers $L \geq L_0$ there are $A_0,h_0,\kappa$ with the following property. For all $A \geq A_0, h \geq h_0$ and $q \in B_{\kappa}(0)$ there is $\rho_0^{\es}>0$ such that if $(H,K)\in \mathbb{D}_k(\rho_0^{\es},\eta,\Lambda)$ then
$$
\Vert \mathbf{S}_k^q(H,K) \Vert_{k+1}^{(A)} \leq \rho^{\es}_0 \eta^{2(k+1)}.
$$
\end{prop}

As the proof will show this estimate reflects the fact that we use first order perturbation: Heuristically, up to first order in $H$,
$$
\mathcal{R}_+ (e^{-H}) \approx e^{-\mathcal{R}_+ H}
$$
since
$
\mathcal{R}_+ (e^{-H}) \approx \mathcal{R}_+ (1 - H)
= 1 - \mathcal{R}_+H \approx e^{-\mathcal{R}_+ H}.
$

\begin{proof}
Fix $\theta < \eta^2$. Let $L_0$ be large enough such that Proposition \ref{Prop:FirstDerivative_BulkFlow} and Proposition \ref{Prop:Smoothness_BulkFlow} can be applied. Define
$$
C_2 = \max \lbrace C_{2,0,0}, C_{1,1,0}, C_{0,2,0} \rbrace
$$
where $C_{j_1,j_2,j_3}$ are the constants from Proposition \ref{Prop:Smoothness_BulkFlow}. Choose $\rho_0^{\es}$ small enough that
$$
\rho_0^{\es} \leq \rho(A)
\quad \text{and }
\theta + 2 C_2 \rho_0^{\es} \leq \eta^2.
$$
Then $(H,K) \in \mathbb{D}_k(\rho_0^{\es},\eta,\Lambda)$ implies $(H,K) \in U_{\rho(A)}$ so we can apply Proposition \ref{Prop:Smoothness_BulkFlow} to estimate as follows.

We Taylor-expand $\mathbf{S}(H,K)$ up to first order with second order integral remainder around $(0,0)$:
\begin{align*}
\mathbf{S}(H,K) &= \mathbf{S}(0,0) + D\mathbf{S}(0,0)(H,K) + \int_0^1 D^2 \mathbf{S}(tH,tK)(H,K)(H,K) (1-t) \de t
\\
&= \mathbf{C}^q K + \int_0^1 D^2 \mathbf{S}(tH,tK)(H,K)(H,K) (1-t) \de t.
\end{align*}
Then we estimate
\begin{align*}
\Vert \mathbf{S}(H,K) \Vert_{k+1}^{(A)}
&\leq \Vert \mathbf{C}^q \Vert \Vert K \Vert_k^{(A)}
+ \frac{1}{2} C_2 \left( 
\Vert H \Vert_{k,0}^2 + 2 \Vert H \Vert_{k,0} \Vert K \Vert_k^{(A)} + \left(\Vert K \Vert_k^{(A)}\right)^2
 \right)
\\
&\leq \rho^{\es}_0 \eta^{2k} \left( \theta + \frac{1}{2} C_2 4 \rho^{\es}_0 \right)
\leq \rho_0^{\es}\eta^{2(k+1)}.
\end{align*}
The last inequality follows by the assumption on $\rho_0^{\es}$. This finishes the proof.
\end{proof}

\subsection{Infinite-volume flow}\label{subsection:Inf_vol_bulk_flow}

\subsubsection{Definition of the infinite-volume flow}

In our context, the renormalisation map $T_k$ is most naturally defined to be a map in finite volume, since a defining property is that is should preserve the circ product under expectation. There is no analogue of this property for infinite volume. Nevertheless, there is a natural definition of a map $(H,K) \mapsto (H_+,K_+)$ which lives on~$\Z^d$ rather than on a torus $\Lambda$, as an appropriate inductive limit of the corresponding maps on the family of all tori. The infinite-volume map has the advantage that it is defined for all scales $k \in \N$, with no restriction due to finite volume. In particular we can study the limit $k \rightarrow \infty$ which we use to apply an implicit function theorem to the dynamical system defined by the RG.

~\\
Let $\mathcal{B}_k(\Z^d)$ be the set of all $k$-blocks in $\Z^d$ and $\mathcal{P}_k(\Z^d)$ be the set of all finite unions of $k$-blocks. Since we are dealing with boxes $\Lambda$ of varying side length $L^N$ let us introduce the notation $N(\Lambda)$ for the exponent describing the side length of the box~$\Lambda$.

~\\
A relevant functional $H \in M_0(\mathcal{B}_k)$ can easily be thought of as an element dependent on a block living in $\Z^d$ instead of $\Lambda$ due to translation invariance. More precisely, given $H \in M_0(\mathcal{B}_k(\Lambda))$, we define $H^{\Z^d}$ on a block $B \in M_0(\Z^d)$ as $H(B)$ for a translation of $B$ to the fundamental domain of $\Lambda$ and suppress the index $\Z^d$ as well as the translation of the block in the notation.

~\\
The irrelevant part is extended as follows.
\begin{defn}
Let $(K^{\Lambda})_{\Lambda}$ be a family of maps which satisfy the $\left(\Z^d\right)$-property.
 For $X \in \mathcal{P}_k(\Z^d)$ choose $\Lambda$ large enough such that $k< N(\Lambda)$ and $\diam (X) \leq \frac{1}{2} \diam(\Lambda)$. Then we define
\begin{align*}
K^{\Z^d}(X) = K^{\Lambda}(X).
\end{align*}
\end{defn}
Here we use that $X \in \mathcal{P}_k(\Z^d)$ has a straight-forward analogon in $\mathcal{P}_k(\Lambda)$ if $\Lambda$ is large enough which we do not record in the notation.

The definition does not depend on the choice of $\Lambda$ owing to the $(\Z^d)$-property required for the family $(K^{\Lambda})_{\Lambda}$.

~\\
Given $(H,K^{\Z^d})$ and the finite-volume maps $\left(\mathbf{S}^{\Lambda}\right)_{\Lambda}$, we define $K_+^{\Z^d}$ as follows.
\begin{defn}
For $U \in \mathcal{P}_{k+1}(\Z^d)$ choose $\Lambda$ large enough such that $k+1 < N(\Lambda)$ and $\diam (U) \leq \frac{1}{2} \diam (\Lambda)$. Then
\begin{align*}
K_+^{\Z^d}\left(H,K^{\Z^d}\right)(U)
= \mathbf{S}^{\Lambda}\left(H,K^{\Lambda} \vert_{U^*}\right).
\end{align*}
\end{defn}
As it is claimed in Proposition \ref{Prop:Smoothness_BulkFlow} the map $\mathbf{S}^{\Lambda}$ satisfies the restriction property and preserves the $(\Z^d)$-property. Moreover, the map $\mathbf{S}^{\Lambda}$ involves integration with respect to $\mu_{k+1}$ of functionals which again only depend on $U^*$ and thus, referring to Remark~\ref{Rem:FRD_Indep_of_N}, the covariance is also independent of the choice of $\Lambda$. So $K_+^{\Z^d}$ is well-defined.

~\\
Defining the relevant flow in infinite volume is straightforward: Fix $B \in \mathcal{B}_{k+1}(\Z^d)$ and $(H,K^{\Z^d})$. Define
$$
H^{\Z^d}_+(B) = \mathbf{A}^q H (B) + \mathbf{B}^q K^{\Z^d} (B).
$$
As before we can skip the index $\Z^d$ on $H$ due to the following reasoning: Let $k<N(\Lambda)$ and $B \in \mathcal{B}_{k+1}(\Z^d)$. Then $B \in \mathcal{B}_{k+1}(\Lambda)$ and for all $b \in \mathcal{B}_k(B)$ it holds that $K^{\Z^d}(b) = K^{\Lambda}(b)$. Thus $H^{\Z^d}_{k+1}(B) = \sum_{b \in \mathcal{B}_{k}(b)} \mathcal{R}_+ H(b) - \Pi_2 \mathcal{R}_+ K^{\Z^d}(b) = \sum_{b \in \mathcal{B}_{k}(b)} \mathcal{R}_+ H(b) - \Pi_2 \mathcal{R}_+ K^{\Lambda}(b) = H_+^{\Lambda}(B)$.

~\\
We just defined the infinite-volume renormalisation map
\begin{align*}
\mathbf{T}^{\Z^d}_k (H_k,K_k^{\Z^d}, q) = (H_{k+1},K_{k+1}^{\Z^d}).
\end{align*}

Now we extend the norms.

There is no need to change the norm for the relevant variable since it does not depend at all on the size of the torus.

For the irrelevant variable let $X \in \mathcal{P}_k^c(\Z^d)$ and choose $\Lambda$ large enough such that $\diam(X) \leq \frac{1}{2} \diam(\Lambda)$. Then $K^{\Z^d}(X) = K^{\Lambda}(X)$ and we can use the same definition as in \cite{ABKM} for
$$
\Vert K^{\Z^d} (X) \Vert_k = \Vert K^{\Lambda} (X) \Vert_k
= \sup_{\varphi \in \mathcal{V}(X^*)}w_k^{-X}(\varphi) \vert K(X,\varphi) \vert_{k,X,T_{\varphi}}
$$
(the weights $w_k$, $w_{k:k+1}$ and $W_k$ do not depend on the size of the torus as long as $X$ is small enough compared to the torus, see Remark \ref{Remark:Weight_Independent_N}).

\subsubsection{Properties of the infinite-volume renormalisation map}

Due to the definition the single step estimates for the map $(H,K^{\Lambda}) \mapsto (H_+,K_+^{\Lambda})$ can be transferred to the infinite-volume flow.

\begin{prop}[Smoothness and contractivity in infinite volume]\label{Prop:Estimates_BulkFlow_InfiniteVolume}
For any $\theta~\in~(0,1)$ there is $L_0$ such that for all integers $L \geq L_0$ (and corresponding $A,h,\kappa$) the following bounds hold independently of $k$ and $N$ for each $q \in B_{\kappa}(0)$:
\begin{align*}
\Vert \mathbf{C}_k^q \Vert \leq \theta,
\quad
\Vert \left( \mathbf{A}^q_k \right)^{-1} \Vert \leq \frac{3}{4},
\quad
\Vert \mathbf{B}_k^q \Vert \leq \frac{1}{3}.
\end{align*}
The derivatives with respect to $q$ are bounded. Moreover, there is $\rho(A)$ such that
$$
\mathbf{S}_k \in C^{\infty}\left(U_{\rho,\kappa},M(\mathcal{P}_{k+1}^c)\right)$$
and
\begin{align*}
\Vert D_1^{j_1} D_2^{j_2} D_3^{j_3} \mathbf{S}_k(H,K,q)(\dot{H}^{j_1},\dot{K}^{j_2},\dot{q}^{j_3}) \Vert^{(A)}_{k+1}
\leq C_{j_1,j_2,j_3} \Vert \dot{H} \Vert_{k,0}^{j_1} \left( \Vert \dot{K} \Vert_k^{(A)} \right)^{j_2} \Vert \dot{q} \Vert^{j_3}.
\end{align*}
\end{prop}

\subsubsection{Existence of the global flow}

\begin{prop}[Existence of the global flow]\label{Prop:FirstIFT}
Fix $\zeta,\eta \in (0,1)$. There is $L_0$ such that for all integers $L \geq L_0$ there is $A_0,h_0$ and $\kappa$ with the following property.
Given $\epsilon > 0$ there exist $\epsilon_1 > 0$ and $\epsilon_2 > 0$ such that for each $(\mathcal{K},\mathcal{H},q) \in B_{\epsilon_1}(0)\times B_{\epsilon_2} (0)\times B_{\kappa}(0) \subset \mathbf{E} \times M_0(\mathcal{B}_0) \times \R^{(d \times m) \times (d \times m)}_{\text{sym}}$ there exists a unique global flow $\left(H_k,K_k^{\Z^d}\right)_{k \in \N}$ such that
\begin{align*}
\Vert H_k \Vert_{k,0} ,\, \left\Vert K^{\Z^d}_k\right\Vert_k^{(A)} \leq \epsilon \eta^k  \quad \text{for all } k \in \N_0,
\end{align*}
with initial condition given by
$$
K^{\Z^d}_0(X,\varphi) = e^{-\mathcal{H}(X,\varphi)} \prod_{x \in X} \mathcal{K}(\nabla \varphi(x))
$$
and
$$
\left(H_{k+1},K_{k+1}^{\Z^d}\right)
= \mathbf{T}_k^{\Z^d} \left(H_k,K_k^{\Z^d},q\right).
$$
Moreover, the flow is smooth in $(\mathcal{K},\mathcal{H},q)$ with bounds on the derivatives which are independent of $N$.

\end{prop}

Proposition \ref{Prop:FirstIFT} implies that for any $(\mathcal{K},\mathcal{H}) \in B_{\epsilon_1} \times B_{\epsilon_2}$ there is $H_0(\mathcal{K},\mathcal{H})$ such that the flow (using the parameter $q(\mathcal{H})$ in the measures) converges to the fixed point of the RG. In our application we require the $q$-component of $H_0$ to correspond to the parameter $q(\mathcal{H})$ in the measure.

\begin{prop}[Global flow with renormalised initial condition]\label{Prop:SecondIFT}
Let
$$
\left(H_k,K_k^{\Z^d}\right)_k = \left(H_k(\mathcal{K},\mathcal{H}),K_k^{\Z^d}(\mathcal{K},\mathcal{H})\right)_k
$$
be the global flow from Proposition \ref{Prop:FirstIFT}. There is $0 < \delta\leq \epsilon_1$ and a smooth map
$$
\hat{\mathcal{H}}: B_{\delta}(0) \subset \mathbf{E} \rightarrow B_{\epsilon_2}(0) \subset M_0(\mathcal{B}_0)
$$
such that
$$
H_0(\hat{\mathcal{H}}(\mathcal{K}), \mathcal{K}) = \hat{\mathcal{H}}(\mathcal{K})
$$
and $q(\hat{\mathcal{H}(\mathcal{K})}) \subset B_{\kappa}(0)$ for all $\mathcal{K} \in B_{\delta}(0)$. Moreover, the derivatives of $\hat{\mathcal{H}}$ can be bounded uniformly in $N$.
\end{prop}

In what follows we will prove Proposition \ref{Prop:FirstIFT} and Proposition \ref{Prop:SecondIFT}. The proofs are very similar to the corresponding proofs in \cite{ABKM}. In fact, here the arguments are slightly easier since we do not have to care about last scale maps due to the change of the finite-range decomposition, see \eqref{eq:dec}. For the sake of completeness we review most of the steps.

The main ingredient is the application of the implicit function theorem.
For the convenience of the reader, we state the implicit function theorem as we will use it in the following.

\begin{thm}[Implicit function theorem]
Let $X,Y,Z$ be Banachspaces, and for $U \subset X, V \subset Y$ open subsets, let $f$ be a $C^p$ Frechet differentiable map $f: U \times V \rightarrow Z$.
If $(x_0,y_0) \in U \times V$, $f(x_0,y_0) = 0$, and $y \mapsto D_2f(x_0,y_0)y$ isomorphism, then there exist a neighbourhood $U_0$ of $x_0$ in $U$ and a Frechet differentiable $C^p$ map $g: U_0 \rightarrow V$ such that $g(x_0)=y_0$ and $f(x,g(x)) = f(x_0,y_0)$ for all $x\in U_0$.
\end{thm}

We give definitions which prepare the proof of Proposition \ref{Prop:FirstIFT}.
Let us set
\begin{align*}
\mathcal{Z}_{\infty}
= \Big\lbrace
Z=(H_0,H_1,K_1,H_2,K_2, \ldots), 
 \,H_k \in M_0(\mathcal{B}_k), K_k \in &M(\mathcal{P}_k^c), 
\\ &
\Vert Z \Vert_{\mathcal{Z}_{\infty}} < \infty
\Big\rbrace
\end{align*}
where
\begin{align*}
\Vert Z \Vert_{\mathcal{Z}_{\infty}}
= \max \left(
\sup_{k \geq 0} \frac{1}{\eta^k} \Vert H_k \Vert_{k,0},
\,\,
\sup_{k \geq 1} \frac{1}{\eta^k} \Vert K_k \Vert_k^{(A)} \right).
\end{align*}
Clearly, $\Vert \cdot \Vert_{\mathcal{Z}_{\infty}}$ is a norm on $\mathcal{Z}_{\infty}$.
We define a dynamical system on $\mathcal{Z}_{\infty}$ as follows:
\begin{align*}
\mathcal{T}: \mathbf{E} \times M(\mathcal{B}_0) \times \mathcal{Z}_{\infty} \rightarrow \mathcal{Z}_{\infty},
\quad
\mathcal{T}(\mathcal{K},\mathcal{H},Z) = \tilde{Z},
\end{align*}
where
\begin{align*}
\tilde{H}_0(\mathcal{K},\mathcal{H},Z)
&= \left( \mathbf{A}_0^{q(\mathcal{H})} \right)^{-1} \left(H_1 - \mathbf{B}_0^{q(\mathcal{H})} \hat{K}_0(\mathcal{K},\mathcal{H}) \right),
\\
\tilde{H}_k(\mathcal{K},\mathcal{H},Z)
&= \left( \mathbf{A}_k^{q(\mathcal{H})} \right)^{-1} \left(H_{k+1} - \mathbf{B}_k^{q(\mathcal{H})} K_k \right),
\quad k \geq 1,
\\
\tilde{K}_{k+1}(\mathcal{K},\mathcal{H},Z) &= \mathbf{S}_k\left(H_k,K_k,q(\mathcal{H})\right), \quad k \geq 1,
\\
\tilde{K}_1(\mathcal{K},\mathcal{H},Z) &= \mathbf{S}_0\left(H_0,\hat{K}_0(\mathcal{K},\mathcal{H})q(\mathcal{H})\right),
\end{align*}
with fixed initial condition
\begin{align*}
\hat{K}_0(\mathcal{K},\mathcal{H})(X,\varphi)
= e^{-\mathcal{H}(X,\varphi)} \prod_{x \in X} \mathcal{K}(\nabla \varphi(x)),
\end{align*}
and $q(\mathcal{H})$ is the projection on the coefficients of the quadratic part of $\mathcal{H}$.

One easily sees that
$$
\mathcal{T}(\mathcal{K},\mathcal{H},Z) = Z
$$ is satisfied if and only if
$$
\mathbf{T}_k(H_k,K_k,q(\mathcal{H})) = (H_{k+1},K_{k+1})
$$ with $K_0 = \hat{K}_0(\mathcal{K},\mathcal{H})$.

Proposition \ref{Prop:FirstIFT} is equivalent to the statement that for sufficiently small $(\mathcal{K},\mathcal{H})$ there is a unique fixed point $\hat{Z}(\mathcal{K},\mathcal{H})$ which depends smoothly on $(\mathcal{K},\mathcal{H})$.

\begin{proof}[Proof of Proposition \ref{Prop:FirstIFT}]
Let $L_0$ (and $A_0,h_0,\kappa$) and $\rho(A)$ be as in Proposition \ref{Prop:Estimates_BulkFlow_InfiniteVolume}.

Let $f: \mathbf{E} \times M(\mathcal{B}_0) \times \mathcal{Z}_{\infty} \rightarrow \mathcal{Z}_{\infty}$ be the map
$$
f(\mathcal{K},\mathcal{H},Z) = \mathcal{T}(\mathcal{K},\mathcal{H},Z) - Z.
$$
We apply the implicit function theorem on $f$. The required assumptions on $f$ are checked below.

It holds that $f(0,0,0)=0$.
To show that $f$ is smooth we have to check that $\mathcal{T}$ is smooth.

\textbf{Claim:} For every triple $(L,h,A)$ which satisfies $L \geq L_0, h \geq h_0(L), A \geq A_0(L)$ there exist constants $\rho_1 >0,\rho_2>0$ such that $\mathcal{T}$ is smooth in
$$
B_{\rho_1}(0) \times B_{\rho_2}(0) \times B_{\rho(A)}(0) \subset \mathbf{E} \times M(\mathcal{B}_0) \times \mathcal{Z}_{\infty},
$$
 i.e., for all $(\mathcal{K},\mathcal{H},Z) \in B_{\rho_1}(0) \times B_{\rho_2}(0) \times B_{\rho(A)}(0)$,
\begin{align*}
\frac{1}{j_1!j_2!j_3!} \Vert
D_{\mathcal{K}}^{j_1} D_{\mathcal{H}}^{j_2} D_{Z}^{j_3}
\mathcal{T}(\mathcal{K},\mathcal{H},Z)
&(\dot{\mathcal{K}}, \ldots,\dot{\mathcal{H}},\ldots,\dot{Z}) \Vert_{\mathcal{Z}_{\infty}}
\\
&\leq C_{j_1,j_2,j_3}(L,h,A) \Vert \dot{\mathcal{K}} \Vert_{\zeta}^{j_1} \Vert \dot{\mathcal{H}} \Vert_{0,0}^{j_2} \Vert \dot{Z} \Vert_{\mathcal{Z}_{\infty}}^{j_3}  .
\end{align*}
Furthermore $q(\mathcal{H}) \in B_{\kappa}(0)$ for all $\mathcal{H} \in B_{\rho_2}(0)$.

\textbf{Proof of the claim:}
We establish smoothness of the coordinate maps for $\tilde{H}_k$ and $\tilde{K}_k$ in a neighbourhood  of the origin. Let $Z \in B_{\rho(A)}(0)$.
\begin{itemize}
	\item Since $\tilde{K}_{k+1}(\mathcal{K},\mathcal{H},Z) = \mathbf{S}_k(H_k,K_k,q(\mathcal{H}))$ for $k \geq 1$, smoothness follows from the smoothness of $\mathbf{S}_k$ in Proposition \ref{Prop:Estimates_BulkFlow_InfiniteVolume}. The proposition can be applied if
	$$
	(H_k,K_k,q(\mathcal{H})) \in U_{\rho(A),\kappa}.
	$$
	Since $Z \in B_{\rho(A)}(0)$ is assumed, $(H_k,K_k) \in U_{\rho(A)}$ is satisfied. Moreover, the map $\mathcal{H} \mapsto q(\mathcal{H})$ is linear and satisfies
	$$
	\left| q(\mathcal{H}) \right| \leq \frac{C}{h^2} \Vert \mathcal{H} \Vert_{0,0}.
	$$
	For $\rho_2$ small enough we thus have $q(\mathcal{H}) \in B_{\kappa}$.
	Bounds on the derivatives of $\tilde{K}_{k+1}$ are obtained as follows. Note that for $k \geq 1$ the function $\tilde{K}_{k+1}$ does not depend on $\mathcal{K}$.
	\begin{align*}
	&\frac{1}{j_2!j_3!} \frac{1}{\eta^{k+1}} \Vert D_{\mathcal{H}}^{j_2} D_{Z}^{j_3} \tilde{K}_{k+1} (\mathcal{K},\mathcal{H},Z)(\dot{\mathcal{H}}, \ldots, \dot{Z}) \Vert_{k+1}^{(A)}
	\\ & \quad\quad\quad\quad\quad\quad
	\leq
	C_{j_2,j_3} \frac{1}{\eta^{k+1}} \left( \Vert \dot{H}_k \Vert_{k,0} + \Vert \dot{K}_k \Vert_k^{(A)} \right)^{j_3} C_{j_2} \Vert \dot{\mathcal{H}} \Vert_{0,0}^{j_2}
	\\
	&\quad\quad\quad\quad\quad\quad
	\leq
	C_{j_2,j_3} \frac{1}{\eta} \Vert \dot{Z} \Vert_k^{j_3} C_{j_2} \Vert \dot{\mathcal{H}} \Vert_{0,0}^{j_2}.
	\end{align*}
	\item
	The smoothness of $\tilde{H}_k$ follows similarly with the help of Proposition \ref{Prop:Estimates_BulkFlow_InfiniteVolume}.
	\item The smoothness of the map $\tilde{K}_1(\mathcal{K},\mathcal{H},Z) = \mathbf{S}_0(\hat{K}_0(\mathcal{K},\mathcal{H}),H_0,q(\mathcal{H}))$ and bounds on the derivatives are done in detail in \cite{ABKM}. Smoothness for $\hat{K}_0$ is proven in Lemma 12.2, and then we apply Proposition \ref{Prop:Estimates_BulkFlow_InfiniteVolume} and chain rule.
\end{itemize}

Now we show that $Z \mapsto D_Zf(0,0) Z$ is an isomorphism. Since
$$
D_Z f(0,0)Z = D_Z \mathcal{T}(0,0,0)Z-Z
$$
 one needs $Z \mapsto \mathcal{T}(\mathcal{K},\mathcal{H},Z)$ to be a contraction at the origin.
From the definition of the maps $\tilde{H}_k$ and $\tilde{K}_k$ and from Proposition \ref{Prop:FirstDerivative_BulkFlow} it follows that
\begin{align*}
\frac{\de \tilde{H}_k}{\de H_{k+1}}
&= \left( \mathbf{A}_k^0 \right)^{-1}
\quad \text{for } k \geq 0,
\\
\frac{\de \tilde{H}_k}{\de K_k}
&= - \left( \mathbf{A}_k^0 \right)^{-1} \mathbf{B}_k^0
\quad \text{for } k \geq 1,
\\
\frac{\de \tilde{K}_{k+1}}{\de K_k}
&= \mathbf{C}_k^0
\quad \text{for } k \geq 1,
\end{align*}
and all other derivatives vanish. Let $Z\in \mathcal{Z}_{\infty}$ satisfy $\Vert Z \Vert_{\mathcal{Z}_{\infty}} \leq 1$. Let us denote by
$$
Z' = \frac{\partial \mathcal{T}(0,0,Z)}{\partial Z} \bigg\vert_{Z=0}Z,
$$
and denote the coordinates of $Z'$ by $H_k'$ and $K_k'$. The bounds on the operators $\left(\mathbf{A}_k^q\right)^{-1}$, $\mathbf{B}_k^q$ and $\mathbf{C}_k^q$ from Proposition \ref{Prop:FirstDerivative_BulkFlow} and $\Vert Z \Vert_{\mathcal{Z}_{\infty}}$ imply that
\begin{align*}
&\Vert H'_0 \Vert_{0,0}
\leq \left\Vert \left( \mathbf{A}_0^0 \right)^{-1} \right\Vert \eta
\leq \frac{3}{4} \eta,
\\
&\eta^{-k} \Vert H_k' \Vert_{k,0}
\leq \eta^{-k} \left\Vert \left(\mathbf{A}_k^0 \right)^{-1} \right\Vert \eta^{k+1} + \eta^{-k} \left\Vert \left( \mathbf{A}_k^0 \right)^{-1} \right\Vert \Vert \mathbf{B}_k^0 \Vert \eta^k
\leq \frac{3}{4} \left( \eta + \frac{1}{3} \right), \, 1 \leq k,
\\
&\eta^{-1} \Vert K_1' \Vert 
= 0,
\\
&\eta^{-k} \Vert K_k' \Vert
 \leq \eta^{-k} \Vert \mathbf{C}_{k-1}^0 \Vert \eta^{k-1}
\leq \frac{\theta}{\eta},\, k \geq 2.
\end{align*}
For $\eta<1$ this implies that
$$
\left\Vert \frac{\partial \mathcal{T}(0,0,Z)}{\partial Z} \bigg\vert_{Z=0} \right\Vert
\leq \varrho < 1.
$$
Thus we can apply the implicit function theorem. It follows that there exist $\epsilon_1$ and $\epsilon_2$ and a smooth function $\hat{Z}: B_{\epsilon_1}(\mathbf{E}) \times B_{\epsilon_2} (M_0(\mathcal{B}_0)) \rightarrow B_{\rho(A)}(\mathcal{Z}_{\infty})$ such that $\hat{Z}(0,0)=0$ and $\mathcal{T}(\mathcal{K},\mathcal{H},\hat{Z}(\mathcal{K},\mathcal{H})) = \hat{Z}(\mathcal{K},\mathcal{H})$ for all $(\mathcal{K},\mathcal{H}) \in B_{\epsilon_1}(0) \times B_{\epsilon_2}(0)$.

~\\
It remains to show that the bounds mentioned in Proposition \ref{Prop:FirstIFT} are satisfied.

The fixed point map satisfies
$$
\Vert\hat{Z}(\mathcal{K},\mathcal{H}) \Vert_{\mathcal{Z}_{\infty}}
\leq \rho(A)
$$
uniformly in $(\mathcal{K},\mathcal{H}) \in B_{\epsilon_1}(0) \times B_{\epsilon_2}(0)$.

The connections between the parameters $\epsilon_1, \epsilon_2$ and $\epsilon$ is clearly explained in \cite{ABKM}.

From this it follows that
$$
\Vert \hat{H}_k \Vert_{k,0} \text{ and } \Vert \hat{K}_k \Vert_{k}^{(A)} \leq \epsilon \eta^k.
$$
\end{proof}

\begin{proof}[Proof of Proposition \ref{Prop:SecondIFT}]
Let $\hat{Z}: B_{\epsilon_1}(0) \times B_{\epsilon_2}(0) \rightarrow B_{\epsilon}(0)$ be the fixed point map from Proposition \ref{Prop:FirstIFT}. Denote by $\Pi_{H_0}: \mathcal{Z}_{\infty} \rightarrow M_0(\mathcal{B}_0)$ the bounded linear map that extracts the coordinate $H_0$ form $Z$.

Define
$$
f(\mathcal{K},\mathcal{H}) = \Pi_{H_0} \hat{Z} (\mathcal{K},\mathcal{H}) - \mathcal{H}
$$
as a map from $B_{\epsilon_1}(0) \times B_{\epsilon_2}(0) \rightarrow M_0(B_0)$. $f$ is surely smooth. The equality
$$
f(0,0) = \Pi_{H_0} \hat{Z}(0,0) = 0
$$ holds since $\hat{Z}(0,0) = 0$.
Our next concern is to show that
$$
D_2f(0,0) \mathcal{H} = -\mathcal{H}.
$$
By definition, $\mathcal{T}(0,\mathcal{H},0)=0$ for all $\mathcal{H} \in B_{\epsilon_2}(0)$. Due to the uniqueness of the fixed point, $\hat{Z}(0,\mathcal{H})=0$ for all $\mathcal{H} \in B_{\epsilon_2}(0)$. It follows that $D_2 \hat{Z}(0,0)=0$ and thus $D_2 \Pi_{H_0}\hat{Z}(0,0) \mathcal{H} = 0$ for all $\mathcal{H} \in B_{\epsilon_2}(0)$.

In summary we obtain that $D_2f(0,0) \mathcal{H}$ is an isomorphism. By the implicit function theorem it follows that there is $\delta$ and a smooth function $\hat{\mathcal{H}}: B_{\delta}(0) \subset \mathbf{E} \rightarrow B_{\epsilon_2}(0) \subset M_0(\mathcal{B}_0)$ such that $\Pi_{H_0} \hat{Z}(\mathcal{K},\hat{\mathcal{H}}(\mathcal{K}))= \hat{\mathcal{H}}(\mathcal{K})$.
\end{proof}

\subsection{Back to finite volume and proof of Theorem \ref{Thm:RepresentationPartitionFunction}}\label{subsec:Back_to_finite-volume}

In the last section we constructed the global flow $(H_k,K_k^{\Z^d})_{k \in \N}$ and proved useful estimates. Now we transfer the properties to the finite-volume flow and deduce the proof of Theorem~\ref{Thm:RepresentationPartitionFunction}.

\subsubsection{Estimates for the finite-volume flow}

The relevant part of the flow is the same in finite and infinite volume,
$$
H_k^{\Z^d} = H_k^{\Lambda}
\quad \text{for } k \leq N(\Lambda),
$$
so the estimates of the global flow are also valid in finite volume.
The irrelevant parts coincide only for polymers $X$ with $\diam(X) \leq \frac{1}{2} \diam(\Lambda)$. However, we can use the improved bound on $D \mathbf{S}_k$ in Lemma \ref{lem:First_der_improved_large_polymers} and the single step estimate in Proposition \ref{Prop:SingleStepRG_BulkFlow} to prove inductively that $K_k^{\Lambda}$ also satisfies the desired estimates.

\begin{prop}[Existence of the finite-volume bulk flow]\label{Prop:Existence_BulkFlow}
Fix $\zeta,\eta \in (0,1)$. There is $L_0$ such that for all integers $L \geq L_0$ there is $A_0,h_0,\kappa$ with the following property. There is $\bar{\delta}$ and $\bar{\epsilon}$ such that for a fixed $\Lambda$ the finite-volume flow
$$
(H_k,K_k^{\Lambda}) \mapsto (H_{k+1},K_{k+1}^{\Lambda})
$$
exists for all $k \leq N(\Lambda)$, is smooth in $\mathcal{K}\in B_{\bar{\delta}}(0)$ with bounds which are uniform in $N(\Lambda)$ and satisfies $(H_k,K_k^{\Lambda}) \in \mathbb{D}_k(\bar{\epsilon},\eta,\Lambda)$.

Moreover,
$$
\Pi_2(H_0(\mathcal{K})) = q(\mathcal{K})
$$
and
$$
K_0(\varphi,X) = K_0(\mathcal{K},H_0)(\varphi,X) = e^{-H_0(\varphi,X)}\prod_{x \in X} \mathcal{K}(\nabla\varphi(x)).
$$
\end{prop}

\begin{proof}
Let $(L_0,A_0,h_0,\kappa)$ be as in Proposition \ref{Prop:FirstIFT} and let $(H_k,K_k^{\Z^d})$ be the global flow with renormalised initial condition from Proposition \ref{Prop:SecondIFT}. Let $\bar{\epsilon} = \min \lbrace \rho_0^{\es}, \epsilon \rbrace$, where $\rho_0^{\es}$ is the quantity from Proposition \ref{Prop:SingleStepRG_BulkFlow} and $\epsilon$ is as in Proposition \ref{Prop:FirstIFT}. 
From the infinite-volume flow we already know that $\Vert H_k \Vert_{k,0} \in B_{\epsilon \eta^k}(0)$ for any $k \leq N$  where $\epsilon$ can be made arbitrarily small by decreasing $\epsilon_1$, in particular we can presume that $\Vert H_k \Vert_{k,0} \in B_{\bar{\epsilon}\eta^k}(0)$ for $\mathcal{K} \in B_{\bar{\delta}}(0)$ for sufficiently small $\bar{\delta}$. Thus we just have to show that $K_k^{\Lambda} \in B_{\bar{\epsilon} \eta^{2k}}(0)$ for $\bar{\delta}$ small enough.

We proceed by induction. For $k=0$ it holds by definition that $K_0^{\Lambda} = K_0^{\Z^d}$ and thus $K_0^{\Lambda} \in B_{\bar{\epsilon}}(0)$ is satisfied.

Now let $K_k^{\Lambda} \in B_{\bar{\epsilon} \eta^{2k}}(0)$ for $ \mathcal{K} \in B_{\bar{\delta}}(0)$. 
 To advance the induction, we apply Proposition~\ref{Prop:SingleStepRG_BulkFlow} and obtain that also $K_{k+1}^{\Lambda}$ satisfies the desired estimate.


~\\
Smoothness in $\mathcal{K}$ can be proven as follows. Let $\mathcal{P}_k^1(\Lambda)$ be the set of polymers $X~\in~\mathcal{P}_k(\Lambda)$ such that $\diam(X) \leq \frac{1}{2} \diam(\Lambda)$. Since $K_k^\Lambda = K_k^{\Z^d}$ on $\mathcal{P}_k^1(\Lambda)$ and since the global flow $(H_k,K_k^{\Z^d})$ is smooth in $\mathcal{K}$, we know that for all $r \in \N$ there is $\tilde{C}_r > 0$ such that for all $k \in \N$
\begin{align}
\left\Vert
D^r_{\mathcal{K}} H_k (\dot{\mathcal{K}}, \ldots, \dot{\mathcal{K}})
\right\Vert_{k,0}
&
\leq
\tilde{C}_r \Vert \dot{\mathcal{K}} \Vert^r_{\zeta},
\label{eq:smoothness_H}
\\
\left\Vert
D^r_{\mathcal{K}} K_k^{\Lambda}\big\vert_{\mathcal{P}_k^1(\Lambda)} (\dot{\mathcal{K}}, \ldots, \dot{\mathcal{K}})
\right\Vert_{k,0}
&
\leq
\tilde{C}_r \Vert \dot{\mathcal{K}} \Vert^r_{\zeta}.
\label{eq:smoothness_K_small}
\end{align}
We will prove inductively (induction on $k$ and $r$) that also $D^r_{\mathcal{K}} K_k^{\Lambda}\big\vert_{\mathcal{P}_k^2(\Lambda)}$ satisfies a bound with a constant $\bar{C}_r$ that is uniform in $k$ and $N$,
\begin{align}
\left\Vert
D^r_{\mathcal{K}} K_k^{\Lambda}\big\vert_{\mathcal{P}_k^2(\Lambda)} (\dot{\mathcal{K}}, \ldots, \dot{\mathcal{K}})
\right\Vert_{k,0}
&
\leq
\bar{C}_r \Vert \dot{\mathcal{K}} \Vert^r_{\zeta}.
\label{eq:smoothness_K_long}
\end{align}
Note that by Lemma \ref{lem:First_der_improved_large_polymers} it holds
\begin{align*}
\left\Vert
D_HD_KD_q \mathbf{S}_k \big\vert_{\mathcal{P}^2_{k+1}(\Lambda)}
(\dot{H},\dot{K})
\right\Vert_{k+1}^{(A)}
\leq
C_1
A^{4-\frac{x}{2}L^{N-(k+1)}}
\Vert\dot{H}\Vert_{k,0}\Vert\dot{K}\Vert_k^{(A)} \Vert \dot{q} \Vert.
\end{align*}
If $k$ is small enough such that for $0 <\vartheta < 1$
$$
4 - \frac{x}{2}L^{N-(k+1)} \leq \vartheta < 0,
\quad
\text{i.e., if} \quad
k< N-M,
\text{ for }
 M= \left\lfloor 1-\frac{\ln(8/x)}{\ln(L)} \right\rfloor,
$$
we can choose $A$ large enough such that for $0<\varrho<1$
$$
C_1 A^{4-\frac{x}{2}L^{N-(k+1)}}
\leq C_1 A^{\vartheta}
\leq \varrho < 1.
$$
In this case, i.e., if $k< N-M$,
\begin{align}
\left\Vert
D_HD_K D_q \mathbf{S}_k(H,K,q)(\dot{H},\dot{K}) \big\vert_{\mathcal{P}_{k+1}^2(\Lambda)}
\right\Vert^{(A)}_{k+1}
\leq
\varrho
\Vert\dot{H}\Vert_{k,0}\Vert\dot{K}\Vert_k^{(A)} \Vert \dot{q} \Vert
\label{eq:improved_DS}
\end{align}
This estimate will be the main point in the argument to advance the induction.
For the remaining scales $k \in \lbrace N-M, N-(M+1), \ldots N \rbrace$ we will use
\begin{align}
\left\Vert
D_HD_K D_q \mathbf{S}_k(H,K,q)(\dot{H},\dot{K}) \big\vert_{\mathcal{P}_{k+1}^2(\Lambda)}
\right\Vert^{(A)}_{k+1}
\leq
C_1
\Vert\dot{H}\Vert_{k,0}\Vert\dot{K}\Vert_k^{(A)} \Vert \dot{q} \Vert
\label{eq:S_high_scale}
\end{align}
where $C_1$ is the constant that appears in Proposition \ref{Prop:Smoothness_BulkFlow}. Accumulation of constants above scale $N-M$ is no problem since there are only finitely many (independently of $k$ and $N$) scales left.

We start with the case $r=1$. We use induction on $k$ until scale $N-M-1$.
Choose
$$
\bar{C}_1 \geq
 \max\left\lbrace
\tilde{C}_1,
 \frac{\varrho}{(1 - \varrho)}3 \tilde{C}_1
\right\rbrace.
$$
 For $k=0$ nothing is to show since both $H_0$ and $K_0^{\Lambda}$ coincide with the corresponding maps in the global flow.
To advance the induction (until scale $N-M-1$), let us assume that
$$
\left\Vert
D_{\mathcal{K}} K_k^{\Lambda} \big\vert_{\mathcal{P}_k^2(\Lambda)} \dot{\mathcal{K}} 
\right\Vert_k^{(A)} 
\leq \bar{C}_1
\Vert\dot{\mathcal{K}} \Vert_{\zeta}.
$$
Then, as long as $k+1 < N-M$, by \eqref{eq:improved_DS}, \eqref{eq:smoothness_H}, \eqref{eq:smoothness_K_small} and induction hypothesis,
\begin{align*}
&
\left\Vert D_{\mathcal{K}} K_{k+1}^{\Lambda} \big\vert_{\mathcal{P}_{k+1}^2(\Lambda)} \dot{\mathcal{K}}
\right\Vert_{k+1}^{(A)}
=
\left\Vert D \mathbf{S}_k (H_k,K_k,q) D_{\mathcal{K}}(H_k,K_k,q) \dot{\mathcal{K}}
\right\Vert_{k+1}^{(A)}
\\
& \quad\quad
\leq
\varrho
\left(
\left\Vert D_{\mathcal{K}} H_k \dot{\mathcal{K}} \right\Vert_{k,0}
+
\left\Vert D_{\mathcal{K}} K_k \big\vert_{\mathcal{P}_k^1(\Lambda)} \dot{\mathcal{K}} \right\Vert_{k}^{(A)}
+
\left\Vert D_{\mathcal{K}} K_k \big\vert_{\mathcal{P}_k^2(\Lambda)} \dot{\mathcal{K}} \right\Vert_{k}^{(A)}
+ 
\left\Vert D_{\mathcal{K}} q \dot{\mathcal{K}} \right\Vert
\right)
\\
& \quad\quad
\leq
\varrho
\left(
\tilde{C}_1 + \tilde{C}_1 + \bar{C}_1 + \tilde{C}_1
\right)
\Vert \dot{\mathcal{K}} \Vert_{\zeta}.
\end{align*}
Our choice of $\bar{C}_1$ and $\varrho<1$ implies that $\varrho\left(\tilde{C}_1 + \tilde{C}_1 + \bar{C}_1\right) \leq \bar{C}_1$ and the induction step is proven. 

If $k= N-M + l$, $l \in \lbrace 0, \ldots, M \rbrace$, then we get inductively by \eqref{eq:S_high_scale}
$$
\left\Vert
D_{\mathcal{K}} K_{k+1}^{\Lambda} \big\vert_{\mathcal{P}^2_{k+1}(\Lambda)} \dot{\mathcal{K}}
\right\Vert_{k+1}^{(A)}
\leq D_{l+1} \Vert \dot{K} \Vert_{\zeta},
$$
where $D_l$ is given by the recursion
$$
D_0 = \bar{C}_1,
\quad
D_l = C_1 \left( 2 \tilde{C}_1 + D_{l-1} \right).
$$
These constants are also independent of $k$ and $N$ since $l$ is independent of $k$ and $N$.

~\\
Next we consider the case $r=2$. Again we use induction on $k$ until scale $N-M-1$. Choose
$$
\bar{C}_2 \geq \max
\left\lbrace
\tilde{C}_2,
\frac{1}{1-\varrho} 
\left(
C_2 (3\tilde{C}_1 + \bar{C}_1)^2 + \varrho 3 \tilde{C}_2
\right)
\right\rbrace,
$$
where $C_2$ is the constant which appears in the estimate
$
\left\Vert
D^2\mathbf{S}_k(H,K)
\right\Vert
\leq C_2
$
in Proposition \ref{Prop:Smoothness_BulkFlow}. For $k=0$ nothing is to show.
Let us assume that
$$
\left\Vert
D_{\mathcal{K}}^2 K_k^{\Lambda}\big\vert_{\mathcal{P}_k^2(\Lambda)} \dot{\mathcal{K}}^2
\right\Vert_k^{(A)}
\leq \bar{C}_2 \Vert\dot{\mathcal{K}}\Vert_{\zeta}^2.
$$
By chain rule we have
\begin{align*}
&
D^2_{\mathcal{K}}
K_{k+1}^{\Lambda} \big\vert_{\mathcal{P}_{k+1}^2(\Lambda)}
\left(\dot{\mathcal{K}},\dot{\mathcal{K}} \right)
\\
&\quad\quad
=
D^2 \mathbf{S}_k (H_k,K_k,q)\big\vert_{\mathcal{P}_{k+1}^2(\Lambda)}
\left(
D_{\mathcal{K}} (H_k,K_k,q) \dot{\mathcal{K}}
\right)^2
\\
& \quad\quad\quad
+
D \mathbf{S}_k (H_k,K_k,q) \big\vert_{\mathcal{P}_{k+1}^2(\Lambda)}
D^2_{\mathcal{K}} (H_k,K_k,q)
\left( \dot{\mathcal{K}},\dot{\mathcal{K}} \right)
\end{align*}
and thus we can estimate with \eqref{eq:improved_DS}, \eqref{eq:smoothness_H}, \eqref{eq:smoothness_K_small} and induction hypothesis,
\begin{align*}
&
\left\Vert
D^2_{\mathcal{K}}
K_{k+1}^{\Lambda} \big\vert_{\mathcal{P}_{k+1}^2(\Lambda)}
\left(\dot{\mathcal{K}},\dot{\mathcal{K}} \right)
\right\Vert_{k+1}^{(A)}
\\ &
\quad\quad\quad
\leq
C_2 \left( \tilde{C}_1 + \tilde{C}_1 + \bar{C}_1 + \tilde{C}_1 \right)^2
+ \varrho \left( \tilde{C}_2 + \tilde{C}_2 + \bar{C}_2 + \tilde{C}_2 \right)
\Vert \dot{\mathcal{K}} \Vert_{\zeta}^2.
\end{align*}
The desired bound is satisfied by our choice of $\bar{C}_2$ and since $\varrho<1$. The key point here is, that the "dangerous" bound $\bar{C}_2$ (the application of the induction hypothesis) comes with the occurrence of $\varrho$.

As before the scales $N-M \leq k \leq N$ can be handled by allowing the constants to accumulate in dependence on $M$. This is no problem since $M$ is independent of $k$ and $N$.

~\\
By a second induction in $r$ we show that \eqref{eq:smoothness_K_long} holds for any $r$.

From the chain rule we deduce inductively that
$$
D^r_{\mathcal{K}} K_{k+1}^{\Lambda} \big\vert_{\mathcal{P}_{k+1}^2(\Lambda)} 
\left( \dot{\mathcal{K}}, \ldots, \dot{\mathcal{K}} \right)
$$
is a linear combination of terms
$$
\left( D^i \mathbf{S}_k(H_k,K_k,q) \right)
\left( D^{j_1}_{\mathcal{K}} (H_k,K_k,q) \dot{\mathcal{K}}^{j_1} \right)
\ldots
\left( D^{j_k}_{\mathcal{K}} (H_k,K_k,q) \dot{\mathcal{K}}^{j_k} \right),
$$
where $1 \leq i \leq r$, $j_s \geq 1$ and $\sum_{s=1}^i j_s = r$.
For $i > 1$ this term is estimated as follows:
\begin{align*}
&
\left\Vert
\left( D^i \mathbf{S}_k(H_k,K_k,q) \right)
\left( D^{j_1}_{\mathcal{K}} (H_k,K_k,q) \dot{\mathcal{K}}^{j_1} \right)
\ldots
\left( D^{j_s}_{\mathcal{K}} (H_k,K_k,q) \dot{\mathcal{K}}^{j_s} \right)
\right\Vert_{k+1}^{(A)}
\\
& \quad\quad\quad
\leq
C_i \prod_{s=1}^i \left( 3\tilde{C}_{j_s} + \bar{C_{j_s}} \right) \Vert \dot{\mathcal{K}} \Vert_{\zeta}^{j_s},
\end{align*}
where we used that $\Vert D^i \mathbf{S}_k (H,K) \Vert \leq C_i$, \eqref{eq:smoothness_H}, \eqref{eq:smoothness_K_small} and induction hypothesis. Note that for $i > 1$ it holds that $j_s < r$ so that only constants $\bar{C}_{l}$ for $l < r$ appear. The term with $i=1$ is
$$
\left(
D \mathbf{S}_k (H_k,K_k,q)
\right)
D^r_{\mathcal{K}}(H_k,K_k,q) \dot{\mathcal{K}}^r,
$$
which can be bounded for scales $k<N-M$ with the help of \eqref{eq:improved_DS} by
$$
\left\Vert
\left(
D \mathbf{S}_k (H_k,K_k,q)
\right)
D^r_{\mathcal{K}}(H_k,K_k,q) \dot{\mathcal{K}}^r
\right\Vert_{k+1}^{(A)}
\leq
\varrho
\left(
3 \tilde{C}_r + \bar{C}_r
\right).
$$
Again the "dangerous" term $\bar{C}_r$ appears with $\varrho$ in front, so that in summary we get
$$
\left\Vert 
D^r_{\mathcal{K}} K^{\Lambda}_{k+1} \big\vert_{\mathcal{P}^2_{k+1}(\Lambda)} \dot{\mathcal{K}}^r
\right\Vert_{k+1}^{(A)}
\leq D + \varrho \left( 3 \tilde{C}_r + \bar{C}_r \right)
$$
for a constant $D$ which depends on $C_i$ for $1 < i \leq r$ and $\tilde{C}_{j_s}$ for $1 \leq j_s < r$.
By the choice
$$
\bar{C}_r \geq \frac{1}{1 - \varrho} \left( D + \varrho 3 \tilde{C}_r \right)
$$
we obtain \eqref{eq:smoothness_K_long}.

Constants are allowed to accumulate for scales $N-M \leq k \leq N$ since $M$ is independent of $k$ and $N$.

~\\
This finishes the proof of smoothness of the finite-volume flow in $\mathcal{K}$.

\end{proof}

\subsubsection{Proof of Theorem \ref{Thm:RepresentationPartitionFunction}}

\begin{proof}[Proof of Theorem \ref{Thm:RepresentationPartitionFunction}]

Let $L_0$ and $\epsilon_0 = \bar{\delta}$ be
 as in Proposition \ref{Prop:Existence_BulkFlow}. Let $f \in \chi_N$.

The starting point is the identity
\begin{align*}
\mathcal{Z}_{N}(\mathcal{K}, \mathcal{Q},f) = \int e^{(f,\varphi)} \sum_{X \subset \T^N} \prod_{x \in X} \mathcal{K} \left(\nabla \varphi(x)\right) \mu_{\mathcal{Q}}(\de\varphi).
\end{align*}
Let us denote
$$
F(\Lambda,\varphi) = \sum_{X \subset \T^N} \prod_{x \in X} \mathcal{K} \left(\nabla \varphi(x)\right).
$$
For $\mathcal{K} \in B_{\epsilon_0}(0)$ let $q=q(\mathcal{K})$ be the quadratic part in $H_0(\mathcal{K})$ from Proposition \ref{Prop:Existence_BulkFlow}.
 Then
\begin{align*}
\mathcal{Z}_{N}(\mathcal{K},\mathcal{Q},f) = \int e^{(f,\varphi)} F(\Lambda,\varphi) \mu_{\mathcal{Q}}(\de\varphi)
&= \frac{Z_N^{(q)}}{Z_N^{(0)}}
\int e^{(f,\varphi)} F^q(\Lambda,\varphi) \mu_{\mathcal{C}^q}(\de\varphi)
\\
&=\frac{Z_N^{(q)}}{Z_N^{(0)}}
e^{\frac{1}{2}(f,\mathcal{C}^q f)}\int F^q(\Lambda,\varphi + \mathcal{C}^qf) \mu_{\mathcal{C}^q}(\de\varphi)
\end{align*}
with 
$$
F^q(\Lambda,\varphi) = e^{\frac{1}{2}\sum_{i,j=1}^d(\nabla_i\varphi,q_{ij}\nabla_j\varphi)} F(\Lambda,\varphi).
$$
 Now let $e=e(\mathcal{K})$ be the constant part and $l(\mathcal{K})(\varphi)$ the linear part of $H_0(\mathcal{K})(\varphi)$. Since $\sum_{x \in \Lambda} l(\mathcal{K})(\varphi)(x)=0$, and since $K_0$ satisfies the correct initial data, it holds that
$$
F^q(\Lambda,\varphi) = e^{-e L^{Nd}} e^{-H_0}\circ K_0(\Lambda,\varphi),
$$
and thus, by Proposition \ref{Prop:Existence_BulkFlow},
\begin{align*}
\mathcal{Z}_{N}(\mathcal{K}, \mathcal{Q},f)
&= \frac{Z_N^{(q(\mathcal{K}))}}{Z_N^{(0)}} e^{\frac{1}{2}(f,\mathcal{C}^{q(\mathcal{K})} f)} e^{-e(\mathcal{K}) L^{Nd}}
\int 
\left(e^{-H_0(\mathcal{K})} \circ K_0(\mathcal{K})\right) (\Lambda,\varphi + \mathcal{C}^{q(\mathcal{K})} f) 
\mu_{\mathcal{C}^{q(\mathcal{K})}}(\de\varphi)
\\
&= \frac{Z_N^{(q(\mathcal{K}))}}{Z_N^{(0)}} e^{\frac{1}{2}(f,\mathcal{C}^{q(\mathcal{K})} f)} e^{-e(\mathcal{K}) L^{Nd}}
\left( e^{-H_N(\mathcal{K})} + K_N(\mathcal{K}) \right)(\Lambda,\mathcal{C}^{q(\mathcal{K})} f).
\end{align*}
Let
$$
Z_N(\mathcal{K},\mathcal{C}^{q(\mathcal{K})} f)
= \left( e^{-H_N(\mathcal{K})} + K_N(\mathcal{K}) \right)(\Lambda,\mathcal{C}^{q(\mathcal{K})} f).
$$
The map $Z_N$ is smooth in $\mathcal{K}$ uniformly in $N$ by Proposition \ref{Prop:Existence_BulkFlow} and and linearity and uniform boundedness of the projection $H_0 \mapsto q(H_0)$.
We shall have established the proof of the theorem if we show that there is a constant $C$ such that $Z_N(\mathcal{K},\mathcal{C}^{q(\mathcal{K})} f)$
satisfies the estimate $\left| Z_N(\mathcal{K},\mathcal{\mathcal{C}}^q f) - 1 \right| \leq C \eta^N$ for special choices of $f$.
First we get
\begin{align*}
\left| Z_N(\mathcal{K},\mathcal{\mathcal{C}}^q f) - 1 \right| 
&\leq 
\left| e^{-H_N(\mathcal{C}^q f)} - 1 \right|
+ \left| K_N(\mathcal{C}^q f) \right|
\\
&\leq \Vert K_N \Vert_N^{(A)} w_N^{\Lambda_N}(\mathcal{C}^q f) A^{-1}
+ \vertiii{e^{-H_N} - 1}_N W_N^{\Lambda_N}(\mathcal{C}^q f ).
\end{align*}
For $f=g_N-c_N$ as given in the assumptions of the theorem it holds that $f \in \chi_N$. Then one can show (see Lemma 5.1 in \cite{Hil16} or the proof of Theorem 2.7 in \cite{ABKM}) that
$$
w_N^{\Lambda_N}(\mathcal{C}^q f),
W_N^{\Lambda_N}(\mathcal{C}^q f)
\leq C
$$
for a constant which is independent of $N$.
Moreover, by Lemma 9.3 in \cite{ABKM}, one can estimate
$$
\vertiii{e^{-H_N} - 1}_N
\leq 8 \Vert H_N \Vert_{N,0}
$$
and since $(H_N,K_N)\in \mathbb{D}_N(\bar{\epsilon},\eta,\Lambda)$ by Proposition \ref{Prop:Existence_BulkFlow} we finally get
$$
\left| Z_N^{\es}(\mathcal{K},\mathcal{\mathcal{C}}^q f) - 1 \right| 
\leq C \eta^N
$$ 
for a constant $C$ which is independent of $N$.
\end{proof}

\bibliography{meinbib}
\bibliographystyle{alpha}

\end{document}